\documentclass[a4paper,10pt]{amsart}

\usepackage[left=2.5cm,right=2.5cm,top=3.4cm,bottom=3.5cm]{geometry}

\usepackage{amsmath}
\usepackage{amssymb}
\usepackage{amsthm}
\usepackage[normalem]{ulem}

\usepackage{pgf,tikz}
\usetikzlibrary{shapes,automata,positioning,arrows,calc}
\usepackage{color}
\usepackage{graphicx}
\usepackage{hyperref}
\usepackage{setspace}
\usepackage{chemfig, chemmacros}
\usepackage[version=3]{mhchem}

\usepackage[sort&compress,comma,square,numbers]{natbib}
\usepackage{booktabs}
\usepackage{enumerate,hyperref}
\usepackage{url}
\usepackage{mathabx}

\newcommand{\R}{\mathbb{R}}

\newcommand{\Z}{{\mathbb Z}}

\newcommand{\support}[1]{\ensuremath{A\left(#1\right)}}
\newcommand{\newton}[1]{\ensuremath{\operatorname{N}(#1)}}
\newcommand{\Vector}[1]{\ensuremath{#1}}
\newcommand{\coef}[2]{\ensuremath{\operatorname{coeff}(#1,#2)}}
\newcommand{\New}{\newton}
\newcommand*{\medcap}{\mathbin{\scalebox{1.5}{\ensuremath{\cap}}}}
\newcommand*{\medcup}{\mathbin{\scalebox{1.5}{\ensuremath{\cup}}}}

\DeclareMathOperator*{\sign}{sign}
\DeclareMathOperator{\conv}{conv}

\renewcommand{\k}{{\kappa}}

\definecolor{NiceBlue}{rgb}{0.2,0.2,0.75}
\newif\ifstartedinmathmode
\newcommand{\struc}[1]{{\relax\ifmmode\startedinmathmodetrue\else\startedinmathmodefalse\fi\color{NiceBlue}{\ifstartedinmathmode #1 \else\textit{#1}\fi}}}

\newtheorem{theorem}{Theorem}
\numberwithin{theorem}{section}
\newtheorem{lemma}[theorem]{Lemma}
\newtheorem{proposition}[theorem]{Proposition}
\newtheorem{corollary}[theorem]{Corollary}
\theoremstyle{definition}
\newtheorem{definition}[theorem]{Definition}
\newtheorem{example}[theorem]{Example}
\newtheorem{remark}[theorem]{Remark}

\begin{document}

	\title{Parameter region for multistationarity in $n-$site phosphorylation networks}
	
	\author[E. Feliu, N. Kaihnsa, T. de Wolff, O. Y\"ur\"uk]{Elisenda Feliu$^{1}$, Nidhi Kaihnsa$^{2}$, Timo de Wolff$^3$, O\u{g}uzhan Y\"ur\"uk$^4$}

	
	
	
	\footnotetext[1]{Department of Mathematical Sciences, University of Copenhagen. efeliu@math.ku.dk}
	
	\footnotetext[2]{Division of Applied Mathematics, Brown University. nidhi\_kaihnsa@brown.edu}

	\footnotetext[3]{Institute of Analysis and Algebra, TU Braunschweig. t.de-wolff@tu-braunschweig.de}
	
	\footnotetext[4]{Chair of Discrete Mathematics/Geometry, Technische Universit\"at Berlin. yuruk@math-tu.berlin.de}
	
	\tikzset{every node/.style={auto}}
	\tikzset{every state/.style={rectangle, minimum size=0pt, draw=none, font=\normalsize}}
	\tikzset{bend angle=15}
	
	\maketitle
	
\begin{abstract}
Multisite phosphorylation is a signaling mechanism well known to give rise to multiple steady states, a property termed multistationarity. When phosphorylation occurs in a sequential and distributive manner, we obtain a family of networks indexed by the number of phosphorylation sites $n$. This work addresses the problem of understanding the parameter region where this family of networks displays multistationarity, by focusing on the projection of this region onto the set of kinetic parameters. The problem is phrased in the context of real algebraic geometry and reduced to studying whether a polynomial, defined as the determinant of a parametric matrix of size three, attains negative values over the positive orthant. The coefficients of the polynomial are functions of the kinetic parameters.
For any $n$,  we provide sufficient conditions for the polynomial to be positive and hence, preclude multistationarity,  and also  sufficient conditions for it to attain negative values and hence, enable multistationarity. These conditions are derived by exploiting the structure of the polynomial, its Newton polytope, and employing circuit polynomials. A relevant consequence of our results is that the set of kinetic parameters that enable or preclude multistationarity are both connected for all $n$.
\end{abstract}

\section{\bf Introduction}

Given a dynamical system, the notion of {\em multistationarity} {refers to} the existence of multiple steady states or equilibrium points. Often in biological processes, multistationarity has been associated with and found relevant for cellular decision making processes \cite{laurent1999,ozbudak2004,Xiong:2003jt}. This has led to an extensive interest in developing methods to identify whether parametric models	arising from reaction networks admit multistationarity, and to determine and understand the parameter region where this occurs, e.g.	\cite{Feinbergss,feliu_newinj,wiuf-feliu,PerezMillan,conradi-PNAS,craciun2008,control,crnttoolbox,FKWY,dickenstein:regions,Conradi_Iosif_Kahle,bates_gunawardena}.
	
In this work, we focus on understanding the rate parameters that allow or restrict the multistationarity in a phosphorylation and dephosphorylation network where a substrate has $n$ phosphorylation sites. 
For two sites, that is, for $n=2$, {we obtain the so-called} dual phosphorylation cycle 
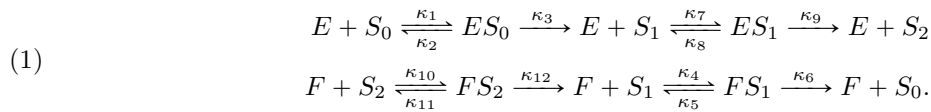
\begin{align}\label{eq:network}
	\begin{split}
	E+S_0 \ce{<=>[\k_1][\k_2]} ES_0 \ce{->[\k_3]} E+S_1   \ce{<=>[\k_7][\k_8]} ES_1 \ce{->[\k_9]} E+S_2 \\
	F+S_2   \ce{<=>[\k_{10}][\k_{11}]} FS_2 \ce{->[\k_{12}]} F+S_1  \ce{<=>[\k_4][\k_5]} FS_1 \ce{->[\k_6]} F+S_0.
	\end{split}
\end{align}
In this system, $S$, $E$, and $F$ denote substrate, kinase, and phosphatase, respectively. Kinase and phosphatase catalyze the phosphorylation and dephosphorylation of the substrate. The subscripts 0, 1, and 2 of the phosphoforms $S_0,S_1$, and $S_2$ denote the number of phosphorylated sites of the substrate $S$. 
		
Over the years, under the assumption of mass-action kinetics, the process of dual phosphorylation has been studied rigorously by means of polynomial ordinary differential equations (ODEs) \cite{dickenstein:regions,FeliuPlos, conradi-mincheva, FKWY}. Besides its relevance in biology, this network has become testing grounds for developing various mathematical methods focused towards understanding multistationarity and dynamics of reaction networks in general. In \cite{FKWY}, the authors studied the region of reaction rate constants (the $\k's$ in \eqref{eq:network}) that yield multistationarity, giving a description of the boundary between the region of  parameters that enable multistationarity and the region of parameters that preclude multistationarity. Moreover, it was shown that both regions are path connected. 
	
\smallskip
In this article we explore the general $n$-phosphorylation network, {\em i.e.},  the substrate has $n$ sites available for phosphorylation, and the network takes the following form:
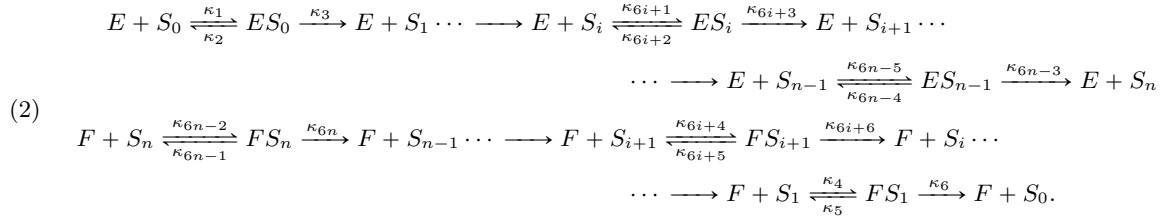
\begin{small}
	\begin{equation}\label{eq:networknew}
		\begin{aligned}
		E+ S_0   \ce{<=>[\k_1][\k_2]} & ES_0  \ce{->[\k_3]} E+ S_1  \cdots \ce{->}E+  S_i   \ce{<=>[\k_{6i+1}][\k_{6i+2}]} ES_i \ce{->[\k_{6i+3}]} E+ S_{i+1}  \cdots  \hspace{1cm} \\ & \hspace{5cm} \cdots \ce{->}E+ S_{n-1}  \ce{<=>[\k_{6n-5}][\k_{6n-4}]} ES_{n-1} \ce{->[\k_{6n-3}]} E+  S_n \\
		F+S_n  \ce{<=>[\k_{6n-2}][\k_{6n-1}]} & FS_n \ce{->[\k_{6n}]} F+S_{n-1} \cdots \ce{->}F+S_{i+1}  \ce{<=>[\k_{6i+4}][\k_{6i+5}]} FS_{i+1} \ce{->[\k_{6i+6}]} F+S_{i}  \cdots  \hspace{1cm} \\&  \hspace{5cm}\cdots \ce{->}F+ S_{1}  \ce{<=>[\k_{4}][\k_{5}]} FS_{1} \ce{->[\k_{6}]} F+S_0.
		\end{aligned}
	\end{equation}
\end{small}%
Multisite phosphorylation, or more general, multisite posttranslational modification, is a ubiquitious process in cell signaling. It is believed that  approximately $30\%$ of all proteins in humans undergo phosphorylation \cite{cohen}, and there are proteins having more than $150$ phosphorylation sites \cite{phosida}.  The order in which the sites are phosphorylated, or how many encounters with the enzyme are required for the phosphorylation of the multiple sites, give rise to different mechanisms \cite{SH09}. Network \eqref{eq:networknew} is a representative model where phosphorylation and dephosphorylation occur sequentially, that is, in a given order, and each encounter leads to the modification of one site, that is, it is distributive. For this mechanism it is known that the network has finitely many positive steady states (in a stoichiometric compatibility class) for any choice of reaction rate constants $\k$, specifically between $1$ and   $2n-1$ \cite{Wang:2008dc}, and that the possible number increases with $n$ \cite{Wang:2008dc,gunawardena:unlimited}. It is conjectured that there exist parameter values for which the system has $2n-1$ positive steady states, but this has only been established for $n\leq 4$ \cite{FHC14}. For general $n$, it has   been proven that there exist parameter values for which the system admits $n+1$ ($n$ even) or $n$ ($n$ odd) positive steady states with half plus one of them asymptotically stable \cite{Wang:2008dc,FRW:stability}. These properties illustrate that the dynamics of the system become qualitatively richer as $n$ increases, but while much is known about the case $n=2$,  very little has been established about the dependence of the dynamics on the choice of parameters for general $n$. 

With this work we address the problem of determining  for which reaction rate constants  $\k_i$, $i \in \{1,\ldots, 6n\}$, the $n$-phos\-pho\-ry\-lation network can have multiple steady states, that is, there exist total amounts of kinase, phosphatase and substrate for which the system has at least two positive steady states. 
In other words, we study the projection of the parameter region of multistationarity onto the subset of parameters consisting of reaction rate constants. 
In \cite{dickenstein:regions} the authors study the projection of this region onto a different subset of parameters, including the three total amounts and some of the reaction rate constants. Their methods are similar in spirit to ours: they rely on connecting properties of the Newton polytope associated to a parametric system of three equations in three variables and certain sign combinations, to obtain a lower bound on the maximal number of positive solutions the system can attain; see also \cite{bihan:santos,giaroli:2n}.

Addressing this question for a fixed $n$ requires challenging computations, since the number of parameters and variables become very large as $n$ increases. However, for any $n$, the differential equation system associated with the network admits exactly three linear first integrals and the set of steady states admits a nice parametrization. These key facts allow us to reduce the study of multistationarity to studying whether a polynomial in three variables attains negative values over the positive orthant  (Proposition~\ref{prop:multi}). This question can in turn be addressed by exploring the Newton polytope associated with the polynomial (Theorem~\ref{thm:NPverts}), and in the end, the multistationarity problem can be reduced to studying the signs that a bivariate polynomial attains (Corollary~\ref{cor:crit}). 
	
With this in place, we proceed similarly to the study of the case $n=2$ in \cite{FKWY}, and relate the signs the polynomial attains to the signs of the coefficients, which depend on the $\k$'s, and to the point configuration of the exponents of the polynomial. 
Imposing that at least one of the coefficients corresponding to vertices of the Newton polytope is negative, we obtain   sufficient conditions for multistationarity (Theorem~\ref{thm:suf}). However, these are not necessary as Theorem~\ref{thm:interiornegative} shows: multistationarity can arise {even} when the conditions in Theorem~\ref{thm:suf} are not satisfied and only points in the interior of the Newton polytope have negative coefficients.
	
To preclude multistationarity, we note that a polynomial can be nonnegative even if some points in the interior of its Newton polytope have negative coefficients. We use SONC (sums of nonnegative circuit) polynomials as a symbolic nonnegativity certificate in order to describe a non-empty region in the parameter space that ensures the monostationarity of the system, {and where the polynomial has some of the coefficients negative} (Theorem~\ref{Proposition:CircuitNonEmptyIntersection}).  

{Our results find subregions of the region of multistationarity and of the region of monostationarity, and indicate that the set of rate parameters that enable multistationarity is not simply found by looking at the sign of the coefficients of the polynomial. For the parameters that fall outside these subregions, our tools cannot characterize what happens, and this asks for the development of new techniques. }

{Even if a full description of the projection of the parameter region of multistationarity of the $n$-phosphorylation network onto the set of reaction rate constants is out of reach, we can still conclude that it is path-connected.  The path-connected components of the parameter region of multistationarity can be thought to represent different biological mechanisms that give rise to multistationarity, as it has recently been brought up in \cite{bates_gunawardena,telek:connectivity}. For example, if the region of multistationarity is connected, then the set of such parameters cannot be classified into two groups corresponding to different mechanisms. This has been proven to be the case for the dual phosphorylation cycle in \cite{telek:connectivity}, but it remains open for $n>2$. 
}

{We study connectivity in Section~\ref{section:connected}, and show in Theorem~\ref{thm:multiconnected} and Corollary~\ref{cor:monoconnected} that the projection of the full parameter region of multistationarity onto the set of reaction rate constants is path connected, and the same holds for the region of monostationarity. To do that, we use the results in the previous sections to identify subregions that are path connected, and then show that any other point is joined to a point in the subregion via a continuous path. 
It is our hope that knowing that these projections are path connected might give a route to show that the full regions are also path connected.   }

Throughout this work, we often view the relevant polynomial, whose coefficients depend on the $\k$'s, as a polynomial in some of the parameters as well, and then employ standard techniques involving the Newton polytope. This is for example the case in showing path connectivity. These ideas should apply to other systems for which the study of multistationarity is reduced to the study of the signs a polynomial attains (those in the setting of \cite{FeliuPlos}). Hence, the methods brought forward in this manuscript may have consequences  beyond the understanding of the system in play.

\section{\textbf{Preliminaries}}\label{section:prelim}	
	
\subsection{\bf The parametric system}

We consider the reaction network \eqref{eq:networknew} and describe the ODE system   that governs evolution of  the concentration of the species in time \cite{FRW:stability,Wang:2008dc,dickenstein:regions}. We then lay the groundwork to derive a polynomial for this system whose positivity will determine the monostationarity of network \eqref{eq:networknew}. 
{We note that network \eqref{eq:networknew}  belongs to a general class of reaction networks called MESSI systems, see \cite{Dickenstein-MESSI}. MESSI systems admit specific decompositions of the sets of species and reactions, and these can be exploited to guarantee that relevant properties, such as absence of boundary steady states or persistence, hold. We will use this fact later on. }

The concentrations of the species in the network are denoted as follows:
	\begin{align*} 
	e=[E],\quad f=[F], \quad  s_i=[S_i]  &\qquad\textrm{for }i=0,\ldots,n, \\ 
	u_{i}=[FS_{i+1}], \quad y_i=[ES_i]  &\qquad\textrm{for } i = 0, \ldots, n-1.
	\end{align*}
Under the assumption of mass-action kinetics, the ODE system in $\R^{3n+3}_{\geq 0}$ is as follows: 
\begin{equation}
		\label{eq:ode}
		\begin{aligned}			 
		\tfrac{de}{dt} &=-\sum_{i=0}^{n-1} \k_{6i+1}\, s_i e +  \sum_{i=0}^{n-1} \k_{6i+1} (\k_{6i+2} + \k_{6i+3})\,  y_i, \\
		\tfrac{df}{dt} & =-\sum_{i=0}^{n-1} \k_{6i+4}\,  s_{i+1} f + \sum_{i=0}^{n-1} (\k_{6i+5} + \k_{6i+6})\, u_i ,\\	
		\tfrac{ds_i}{dt} &= - \k_{6i+1}\, s_i e \k_{6i-3}\,  y_{i-1} + \k_{6i+2}\,  y_i  - \k_{6i-2}\,  s_i f + \k_{6i+6}\,  u_i+ \k_{6i-1}\,  u_{i-1},&& \text{for }i=0,\dots,n, \\ 
		\tfrac{dy_i}{dt} & = \k_{6i+1}\, s_i e - (\k_{6i+2} + \k_{6i+3})\,  y_i, &&\text{for }i = 0,\dots,n-1,   \\
		\tfrac{du_i}{dt} &= \k_{6i+4}\,  s_{i+1} f - (\k_{6i+5} + \k_{6i+6})\, u_i,&& \text{for }i=0,\dots,n-1, \\
		\end{aligned}
		\end{equation}
with the convention that $\k_j=0$ if $j>6n$ or $j<0$ (this becomes relevant only for $\tfrac{ds_0}{dt}$ and $\tfrac{ds_n}{dt}$). This system of ODEs admits three linear first integrals, such that the trajectories are confined in the following level sets for some $E_{\rm tot}, F_{\rm tot}, S_{\rm tot}\geq 0$: 
\begin{equation}\label{eq:cons_laws}
	e+ \sum_{i=0}^{n-1}y_i = E_{\rm tot},\qquad f+\sum_{i=0}^{n-1}u_i=F_{\rm tot},\qquad  s_0+\sum_{i=1}^{n}s_i+\sum_{i=0}^{n-1}y_i+\sum_{i=0}^{n-1}u_i= S_{\rm tot}.
\end{equation}
{It is straightforward to verify that these equations arise from  linear first integrals, and there are no additional independent linear first integrals, as the rank of the coefficient matrix of the polynomials on the right-hand side of \eqref{eq:ode} is three for all $\k$'s; alternatively, see e.g. \cite[Theorem 3.2]{Dickenstein-MESSI}.}
Each of the equations in \eqref{eq:cons_laws} is called a \emph{conservation law} and $E_{\rm tot}, F_{\rm tot}$ and $S_{\rm tot}$ are the \emph{total amounts} of $E,F,$ and $S$, respectively. These are, therefore, taken nonnegative. The intersection of a level set with the nonnegative orthant is called a \emph{stoichiometric compatibility class}. In particular,  the ODE system has $3n+3$ variables and for a given initial condition, the trajectory lies in a $3n$ dimensional subspace determined by the three linear constraints in \eqref{eq:cons_laws}.

The steady states of the system are found by setting the left-hand side of \eqref{eq:ode} to zero. Using the notation in \cite{dickenstein:regions}, we  consider the inverses of Michaelis-Menten constants of the (de)phosphorylation events
\begin{equation}\label{eq:KL}
	K_i=\frac{\k_{6i+1}}{\k_{6i+2}+\k_{6i+3}}, \qquad  L_i=\frac{\k_{6i+4}}{\k_{6i+5}+\k_{6i+6}},\quad \quad i=0,\ldots, n-1, 
\end{equation}
and define 
\begin{equation}\label{eq:T}
	T_i=\prod_{j=0}^{i}\frac{\k_{6j+3}K_j}{\k_{6j+6}L_j} =T_{i-1}\, \frac{\k_{6i+3}K_i}{\k_{6i+6}L_i} , \quad i=0,\ldots, n-1, \qquad T_{-1}=1.
\end{equation}
{We note that here, the definition of $K_i$  is the inverse of that given in \cite{FKWY} for $n=2$, but it turns out to be more convenient to use the notation of \cite{dickenstein:regions} for general $n$.} For later use, we define  the following composition of surjective maps $\eta \colon \R^{6n}_{>0}    \rightarrow \R^{4n}_{>0} \rightarrow  \R^{3n}_{>0}$
\begin{equation}\begin{aligned}\label{eq:assembly}
\kappa = (\k_1,\dots,\k_{6n}) &\mapsto  (\k_3,\k_6,\ldots,\k_{6n}, K_0,\ldots,K_{n-1},L_0,\ldots,L_{n-1})  \\ & \mapsto 
 (T_0,\ldots,T_{n-1}, K_0,\ldots,K_{n-1},L_0,\ldots,L_{n-1}). 
\end{aligned}
\end{equation}

As a consequence of the three conservation laws, the equations corresponding to $e$, $f$, and $s_0$ are redundant as they are linearly dependent on the rest. The steady state equations for $s_{i+1},y_i,u_i$ with $i=0,\dots,n-1$ are linear in these variables, and have a unique solution, which is positive provided $e,f,s_0$ are positive:
\begin{align*}	 
	s_{i+1}  & = T_{i} e^{i+1} f^{-(i+1)} s_0 & 
	y_i &  = K_i T_{i-1} e^{i+1} f^{-i} s_0 \label{eq:solsystem} &
	u_i &  =L_i T_{i} e^{i+1} f^{-i} s_0. 
\end{align*}
This gives rise to a parametrization of the set of positive steady states in the variables $e,f,s_0$,
\begin{eqnarray}\label{eq:param}
	\varphi_\k \colon \R^3_{>0} & \longrightarrow & \R^{3n+3}_{>0},
\end{eqnarray}
that is, the image of $\varphi_\k$ is precisely the set of positive steady states. 
For later use, we consider {the polynomial function
\begin{equation}\label{eq:psi}
	\psi_\k\colon \R^{3n+3} \longrightarrow \R^{3n+3},
\end{equation}
whose first three entries are $(e+ \sum_{i=0}^{n-1}y_i , f+\sum_{i=0}^{n-1}u_i,  s_0+\sum_{i=1}^{n}s_i+\sum_{i=0}^{n-1}y_i+\sum_{i=0}^{n-1}u_i)$, that is, the left-hand side of the three conservation laws listed in \eqref{eq:cons_laws}, and the entries $4$ to $3n+3$ are the right-hand side of the equations for $\tfrac{ds_i}{dt}$ for $i=1,\dots,n$, $\tfrac{dy_i}{dt}$ for $i=0,\dots,n-1$, and $\tfrac{du_i}{dt}$ for $i=0,\dots,n-1$ in \eqref{eq:ode}, in this order}.

In a fixed stoichiometric compatibility class with positive total amounts, this system has finitely many positive steady states for any choice of reaction rate constants $\k$, specifically between $1$ and   $2n-1$ \cite{Wang:2008dc}, see the Introduction. 
	
For a given vector of reaction rate constants $\k$, if there exist positive $E_{tot}, F_{tot},$ and $S_{tot}$ such that some stoichiometric compatibility class has at least two positive steady states, then we say that $\k$  {\em enables} multistationarity. If this is not the case, then $\k$ is said to {\em preclude} multistationarity. For $n=2$, the set of reaction rate constants that enable multistationarity was explored in detail in \cite{FKWY}, building on preliminary results by Conradi and Mincheva~\cite{conradi-mincheva}.	In this work we go beyond the case $n=2$ and determine subsets of reaction rate constants that enable or preclude multistationarity for  general $n$, that is, for the network \eqref{eq:networknew}.

\subsection{Nonnegative Circuit Polynomials}
\label{subsection:PreliminariesSONC}
		
Let $f$ be a polynomial function on $\R^m$. Determining if $f(x)$ is nonnegative over $\R^m$ is a classical problem in real algebraic geometry that appears  naturally  in different contexts, and prominently in polynomial optimization. It has been studied since the 19th century and is the subject of Hilbert's 17th problem \cite{Hilbert:Seminal}. The critical result for the present work, Theorem~\ref{thm:nonneg} below, builds on  pioneering by  Reznick   in 1989 \cite{Reznick:AGI}  where the theorem was proven for a special case. Another special case was shown by Fidalgo and Kovacec in \cite{Fidalgo:Kovacec}. In its full generality, Theorem~\ref{thm:nonneg} was first shown in 2012 by Pantea, Craciun and K\"oppl in  \cite{pantea-jac}. That approach is particularly interesting, since the result is motivated by the study of reaction networks rather than nonnegativity of real polynomials. In this article, we use Theorem~\ref{thm:nonneg} in the terminology  introduced by Iliman and de Wolff in  2016 \cite{Iliman:deWolff:Circuits}. Specifically, we certify nonnegativity by decomposing $f(x)$ as a \emph{sum of nonnegative circuit polynomials} (SONCs), and we use that nonnegative circuit polynomials are characterized by means of Theorem~\ref{thm:nonneg}. The theory of SONCs has been subsequently developed further in \cite{Dressler:Iliman:deWolff:Positivstellensatz,Forsgaard:deWolff:BoundarySONCCone,Iliman:deWolff:GP}.

For a multivariate polynomial $f(x)=\Sigma_{\alpha}c_{\alpha}x_1^{\alpha_1} \cdots x_m^{\alpha_m}$ with $\alpha=(\alpha_1, \ldots,\alpha_m)\in {\Z_{\geq 0}^m}$, the {\em Newton polytope}, denoted by $\newton{f}$, is the convex hull of the exponent vectors of the monomials with non-zero coefficients, i.e. $\newton{f}:=\conv(\{\alpha\in {\Z_{\geq 0}^m} : c_\alpha \neq 0\})$. Often we will also use $\coef{f}{\alpha}$ to denote the {\em coefficient of $x^{\alpha}$ in the polynomial $f(x)$}. We next define circuit polynomials.

\begin{definition}\label{Definition:CircuitPolynomial}
A polynomial $f \in \R[x_1,\dots,x_m]$ is called a \emph{circuit polynomial} if 
\begin{align*}
	f(x) =  c_{\beta} x^{{\beta}} + \sum_{j=0}^r c_{\alpha(j)} x^{\alpha(j)}, 
\end{align*}
where $\beta, \alpha(j) \in {\Z_{\geq 0}^m}$, $c_{\alpha(j)} \in \R_{> 0}$ for $j=0,\dots,r$, $c_{\beta} \in \R$, and $\New{f}$ is a simplex with vertices $\alpha(0), \ldots,\alpha(r)$ such that the exponent $\beta$ is in the strict interior of $\New{f}$. That is, $\beta=\sum_{j=0}^r\lambda_j\alpha(j)$ with $\lambda_j > 0$ and $\sum_{j=0}^r\lambda_j=1.$ We define the \emph{circuit number} as
\begin{equation*}\label{Equ:DefCircuitNumber} 
	\Theta_f \ := \ \prod_{j = 0}^r \left(\frac{c_{\alpha(j)}}{\lambda_j}\right)^{\lambda_j}. 
\end{equation*}
\end{definition}	
	
Note that the circuit number $\Theta_f$ can be computed by solving a system of linear equations. 
In this article, we will consider circuit polynomials restricted to the nonnegative orthant, as the variables will be concentrations  and reaction rate constants. We are interested in polynomials that take nonnegative (resp. positive) values over the positive orthant, which we call respectively \emph{nonnegative} (resp. \emph{positive}) polynomials. Thus, as opposed to the  definition of circuit polynomials given in \cite{Iliman:deWolff:Circuits} where the focus was on nonnegativity over the whole real space, we do not restrict $\alpha(j)$ in Definition \ref{Definition:CircuitPolynomial} to have even entries. We can do this because in the positive orthant, we can redefine the polynomial by replacing every variable with its square. 
Then, the theorem below is a consequence of \cite[Thm. 3.6. and Prop. 3]{pantea-jac} and a direct specialization of Theorem 3.8 in \cite{Iliman:deWolff:Circuits}.

\begin{theorem}\label{thm:nonneg}
With the notation in Definition~\ref{Definition:CircuitPolynomial}, a circuit polynomial $f$ is nonnegative on $\R^m_{\geq 0}$ if and only if $c_{\beta} \geq -\Theta_f.$
\label{Theorem:CircuitPolynomialNonnegativity}
\end{theorem}

The \emph{Motzkin polynomial}, $x^4y^2 + x^2y^4 - 3x^2y^2 + 1 $, is a classical example of a nonnegative polynomial which is not a sum of squares \cite{Motzkin:AMGMIneq}. This is a circuit polynomial with $\beta = (2,2)$,  and $\lambda_j=1/3$ for all $j$ (see  e.g., \cite{Iliman:deWolff:Circuits}). The circuit number is $3$ and hence, the polynomial is nonnegative by Theorem~\ref{thm:nonneg}.

\smallskip

We note that the scenario studied by Reznick in \cite{Reznick:AGI} was the case where $c_{\alpha(j)} = \lambda_j$ for all $j \in \{1,\ldots,n\}$.
In \cite{pantea-jac}, the authors do not only consider the case of simplex Newton polytopes/circuits, but discuss also the case of non-unique barycentric coordinates and give nonnegativity certificates. These functions were later  called AGE (Arithmetic-Geometric Exponential) functions by Chandrasekaran and Shah \cite{Chandrasekaran:Shah:RelativeEntropy} in the context of signomial programming, and a theorem similar to  Theorem \ref{thm:nonneg} was given  in the framework called SAGE.
While it was not obvious that SONC and SAGE describe the same cone of functions, we know nowadays that this is the case. We summarize this fact  in the following remark for later reference.
\color{black}

\begin{remark}
	\label{remark:SAGE}
Let $f$ be a polynomial whose support contains only a single exponent that is not a vertex of $\newton{f}$, and let the coefficients of $f$ at the vertices of $\newton{f}$ be positive. Then $f$ is nonnegative if and only if it is SONC, see \cite[Theorem 3.10]{wang2018nonnegative}, also \cite[Theorem 11]{Chandrasekaran:Murray:Wiermann}. Such polynomials are examples of AGE as defined in \cite{Chandrasekaran:Shah:RelativeEntropy}, and they yield an equivalent description of SONC polynomials.
\end{remark}

We also point out that computationally a SONC/SAGE certificate can be found using relative entropy programs (REP).
This is a convex optimization program, where the target function and the constraints are linear or a sum of entropy functions.
In general, these problems can be solved effectively, for example using standard solvers like MOSEK; for further details on REPs see e.g. \cite{Chandrasekaran:Shah:RelativeEntropyApplications,Chandrasekaran:Murray:Wiermann:REP}.
In the context of reaction networks the situation is more involved though, as the polynomials, for which we aim to certify nonnegativity, come with \textit{symbolic} coefficients,  and this  does not allow  to simply compute a SONC/SAGE certificate via solving an REP in the usual way; see also Remark~\ref{Remark:REPandSAGE}.

%
	
	\section{\bf A polynomial for multistationarity and the Newton polytope}
	
	\label{section:NP}
We follow the approach in \cite{FKWY} and study the space of reaction rate constants that enable multistationarity after first reducing the problem to understanding the positivity of a multivariate polynomial for given reaction rates. This section is devoted to computing the general expression of the   polynomial and exploring the structure of the associated Newton polytope to determine whether it is positive on the positive orthant. To this end, we apply Theorem 1 in \cite{FeliuPlos}, and obtain the following proposition. Here $ J_{\psi_\k}$ denotes the Jacobian matrix of the polynomial function defined in \eqref{eq:psi} in $3n+3$ variables.
	
	\begin{proposition}\label{prop:onepolynomial}
		Let $\k\in \R^{6n}_{>0}$, $\psi_\k$ be the polynomial function   in \eqref{eq:psi} and $\varphi_\k$ the parametrization of the set of positive steady states from \eqref{eq:param}.  Define the following rational function
		\[q_\k(e,f,s_0) := (-1)^{3n}  \det J_{\psi_\k} (\varphi_\k(e,f,s_0)). \]
		The following holds:
		\begin{enumerate}[(i)]
			\item If $q_\k(e,f,s_0)>0$ for all $(e,f,s_0)\in \R^3_{>0}$, then $\k$ precludes multistationarity.
			\item If $q_\k(e^*,f^*,s_0^*)<0$ for some $(e^*,f^*,s_0^*)\in \R^3_{>0}$, then $\k$ enables multistationarity. In this case the stoichiometric compatibility class containing $\varphi_\k(e^*,f^*,s_0^*)$ contains at least two positive steady states. 
		\end{enumerate}
	\end{proposition}
	\begin{proof}
		The statement is exactly the conclusion of Theorem 1 in \cite{FeliuPlos}, as the dimension of the stoichiometric compatibility classes is $3n$. So we need to check that the two hypotheses to apply  Theorem 1 in \cite{FeliuPlos} hold. First, the   reaction network \eqref{eq:network} is dissipative, as it is conservative (each concentration appears in at least one conservation law that has all coefficients nonnegative).  Second, we have to verify that {the ODE system \eqref{eq:ode} has }no boundary steady states (steady states with some entry equal to zero) when total amounts are positive. This follows from Theorem 3.15 in \cite{Dickenstein-MESSI}{, as network~\eqref{eq:network} is a MESSI system that satisfies the theorem's hypotheses.}
	\end{proof}

	Based on Proposition~\ref{prop:onepolynomial}, a strategy to determine the set of reaction rate constants that enable 
	multistationarity consists of studying the signs $q_\k$ attains on the positive orthant. To compute $q_\k$, for all $n$, 
	the first step is to reduce the computation of the determinant of a matrix of size $3n+3$, to that of the determinant of a $3\times 3$ matrix. 
	Let $\Phi$ be the linear map given by the left-hand side of \eqref{eq:cons_laws}.  Then the composition $\Phi\circ \varphi_\k$ is a function from $\R^3_{>0}$ to $\R^3_{>0}$.
	
	\begin{proposition}\label{prop:reduce}
		For any $\k\in \R^{6n}_{>0}$ and $(e,f,s_0)\in \R^3_{>0}$, the sign of $q_\k(e,f,s_0)$ agrees with the sign of
		\[   \det J_{\Phi\circ \varphi_\k}(e,f,s_0). \] 
	\end{proposition}
	\begin{proof}
		By the chain rule, it holds
		\begin{equation}\label{eq:decomposeJac}
		J_{\psi_\k\circ \varphi_\k}(e,f,s_0)= J_{\psi_\k} (\varphi_\k(e,f,s_0)) \cdot J_{\varphi_\k}(e,f,s_0).
		\end{equation}
		We write  $J_{\psi_\k} (\varphi_\k(e,f,s_0))$ in block form as
		\[ J_{\psi_\k} (\varphi_\k(e,f,s_0)) = \begin{pmatrix} A' & A'' \\ B' & B'' \end{pmatrix}, \quad A'\in \R^{3\times 3},\quad B''\in \R^{3n\times 3n}. \]
		The first three components of $\psi_\k \circ \varphi_\k$ agree with $\Phi\circ \varphi_\k$ by construction, and the last $3n$ components are identically zero, as $\varphi_\k$ is precisely found by solving the system given after setting the last $3n$ entries of $\psi_\k$ to zero. Since $\varphi_\k$ is the identity in the first three entries, for some $D\in \R^{3n\times 3}$ we write \eqref{eq:decomposeJac} in block form and obtain the consequent relations:
\begin{align*}
\begin{pmatrix}
		J_{\Phi\circ \varphi_\k}(e,f,s_0) \\ 0 
		\end{pmatrix} &=  \begin{pmatrix} A' & A'' \\ B' & B'' \end{pmatrix} \begin{pmatrix}Id_{3\times 3} \\ D \end{pmatrix} \\[5pt]
J_{\Phi\circ \varphi_\k}(e,f,s_0)  &= A' + A''D,\qquad B'+B''D=0.
\end{align*}
		By adding to the first three columns of $ J_{\psi_\k} (\varphi_\k(e,f,s_0)) $  linear combinations of the remaining $3n$ columns with coefficients given by the three columns of $D$ respectively, we obtain
		\begin{align*}
		\det J_{\psi_\k} (\varphi_\k(e,f,s_0)) & = \det\begin{pmatrix} A' & A'' \\ B' & B'' \end{pmatrix}
		= \det \begin{pmatrix} A' + A''D & A'' \\ B' + B''D & B'' \end{pmatrix}=  \det \begin{pmatrix} J_{\Phi\circ \varphi_\k}(e,f,s_0)   & A''    \\ 0 & B'' \end{pmatrix} \\ & = \det J_{\Phi\circ \varphi_\k}(e,f,s_0)   \cdot \det B''.
		\end{align*}

		All that remains is to show that the sign of $\det B''$ is $(-1)^{3n}$. The matrix $B''$ is the Jacobian of the last $3n$ rows of $\psi_\k$ in the variables $s_i,y_i,u_i$. Note that it is precisely the coefficient matrix of the linear system that is solved to find $\varphi_\k$. 
		To find the determinant, we use the theory developed in \cite{saez:linearalgebra,Fel_elim}, see also \cite{TG-rational}, where this determinant is found using the Matrix-Tree theorem on a suitable digraph. Consider the following digraph:

		\vspace{0.5cm}
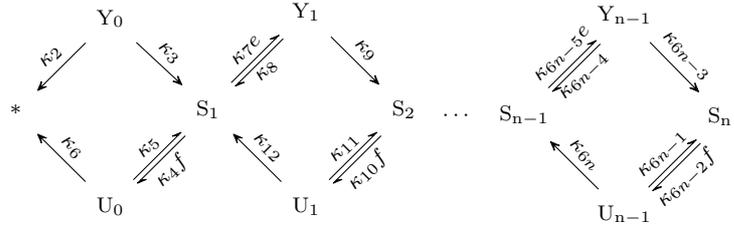
		\begin{footnotesize}
			\begin{tabular}{ccccccccc}
				& & & & & &
				\schemestart
				\subscheme{*}
				\arrow(*--y0){<-[$\kappa_2$]}[45]
				\chemfig{Y_0}
				\arrow(@*--u0){<-[$\kappa_6$]}[315]
				\chemfig{U_0}
				\arrow(@u0--s1){<=>[$\kappa_5$][$\kappa_{4}f$]}[45]
				\chemfig{S_1}
				\arrow(@s1--@y0){<-[$\kappa_3$]}
				\arrow(@s1--y1){<=>[$\kappa_{7}e$][$\kappa_8$]}[45]
				\chemfig{Y_1}
				\arrow(@s1--u1){<-[$\kappa_{12}$]}[315]
				\chemfig{U_1}
				\arrow(@u1--s2){<=>[$\kappa_{11}$][$\kappa_{10}f$]}[45]
				\arrow(@s2--@y1){<-[$\kappa_9$]}
				\chemfig{S_2}
				\schemestop
				&
				$\cdots$
				&
				\schemestart
				\chemfig{S_{n-1}}
				\arrow(sn1--yn1){<=>[$\kappa_{6n-5}e$][$\kappa_{6n-4}$]}[45]
				\chemfig{Y_{n-1}}
				\arrow(@sn1--un1){<-[$\kappa_{6n}$]}[315]
				\chemfig{U_{n-1}}
				\arrow(@un1--sn){<=>[$\kappa_{6n-1}$][$\kappa_{6n-2}f$]}[45]
				\chemfig{S_n.}
				\arrow(@sn--@yn1){<-[$\kappa_{6n-3}$]}
				\schemestop
			\end{tabular}
		\end{footnotesize}
		
		\vspace{0.5cm}

		Consider the directed spanning trees rooted at the vertex $*$ (that is, such that $*$ is the only vertex without outgoing edges), and for each such tree, define its label as the product of the labels of its edges. 
		By \cite[Proposition 1]{saez:linearalgebra}, the determinant of $B''$ is $(-1)^{3n}$ times the sum of the labels of all spanning trees rooted at $*$. As there is at least one such tree, and all labels are positive, this gives that the sign of $\det B''$ is $(-1)^{3n}$.  This concludes the proof.
	\end{proof}

	Proposition~\ref{prop:reduce} is the key in this work, as it will allow us to find explicitly a polynomial whose signs determine whether or not a vector of reaction rate constants enables multistationarity. The proposition relies on some algebraic manipulations and the multivariate chain rule, and this would apply to any other system falling in the setting of Theorem 1 in \cite{FeliuPlos} and for which a parametrization of the set of positive steady states exists. In the notation of the proof of Proposition~\ref{prop:reduce}, a critical aspect is that the sign of the determinant of the matrix $B''$ is constant and can be determined. This might seem very restrictive, but it occurs whenever the parametrization of the set of steady states arises from linear elimination of variables in the setting of \cite{saez:linearalgebra,Fel_elim}, see also \cite{TG-rational}. 
		\color{black}

	\smallskip
	In view of Propositions~\ref{prop:onepolynomial} and \ref{prop:reduce}, 
	multistationarity is established by considering the sign of  the determinant of the Jacobian of $\Phi\circ \varphi_\k$. 
	The entries of $\Phi\circ \varphi_\k$ are
	{\small \begin{align*} 
		e + \sum_{i=0}^{n-1} K_i T_{i-1} e^{i+1} f^{-i} s_0, &   \nonumber \\
		f + \sum_{i=0}^{n-1} L_i T_{i} e^{i+1} f^{-i} s_0, &  \label{eqn:constraints}\\
		s_0 + \sum_{i=1}^{n} T_{i-1} e^i f^{-i} s_0 + \sum_{i=0}^{n-1}(K_i T_{i-1} +L_i T_{i})  e^{i+1} f^{-i} s_0 &. \nonumber
		\end{align*} }%
	Hence, $\Phi\circ \varphi_\k$  depends on $\k$ through the assembled parameters $K_i,L_i,T_i$. 
	As we are only interested in  the determinant of the Jacobian of $\Phi\circ \varphi_\k$, we subtract the first and second equations from the third, and in this way the last equation is replaced by
	\[   -e -f+ s_0 + \sum_{i=0}^{n-1} T_{i} e^{i+1} f^{-(i+1)} s_0. \] 
	
	For convenience, the parameters and variables are relabeled in the following way:
	\begin{equation}\label{eq:abc}
	\begin{aligned}
a_i &= K_i T_{i-1}, & b_i&= L_i T_{i}, & c_i& =T_{i}, & \text{for }i=0,\dots,n-1, \\
x_1 &=e, & x_2& =\tfrac{e}{f}, &  x_3& =s_0.
\end{aligned}
	\end{equation}
With this notation, an easy computation now gives that the determinant of the Jacobian of $\Phi\circ \varphi_\k$ is  the determinant of the following matrix:
		\begin{equation} \label{matrix:Jac2}  {
			J:=	\begin{bmatrix}
			1+  \sum_{i=0}^{n-1} (i+1)\, a_i \, x_2^i x_3& - \sum_{i=0}^{n-1} i\, a_i \, x_2^{i+1}  x_3&\sum_{i=0}^{n-1}  a_i  \, x_1x_2^i \\[5pt] 
			\sum_{i=0}^{n-1} (i+1)\, b_i  \,x_2^i x_3 & 1 - \sum_{i=0}^{n-1} i \,b_i\,  x_2^{i+1} x_3&\sum_{i=0}^{n-1}  b_i  \, x_1 x_2^{i} \\[5pt] 
			-1+\sum_{i=0}^{n-1} (i+1) \, c_i \,  x_1^{-1} x_2^{i+1}  x_3 & -1 - \sum_{i=0}^{n-1} (i+1) \, c_i \, x_1^{-1} x_2^{i+2}  x_3  & 1+\sum_{i=0}^{n-1} c_i \, x_2^{i+1} 
			\end{bmatrix}.
		}
		\end{equation}
	The determinant of \eqref{matrix:Jac2} is a polynomial in $x_1, x_2,x_3$ with coefficients depending on $a_i,b_i,c_i$, which in turn depend on $T_i,K_i,L_i$, $i=0,\dots,n-1$. We, therefore, view $\det J$ as a polynomial in $x_1,x_2,x_3$ with coefficients depending on a parameter vector in $\R^{3n}_{>0}$, and define:
	\begin{equation}
\label{eq:p}
	p_\eta(x_1,x_2,x_3):= \det J,\qquad	\eta = (T_0,\dots,T_{n-1}, K_0,\ldots,K_{n-1},L_0,\ldots,L_{n-1}) \in \R^{3n}_{>0}.
	\end{equation}

	The  discussion above together with Propositions~\ref{prop:onepolynomial} and \ref{prop:reduce}, 
	shows that whether or not $\k$ enables multistationarity depends only on the entries of the vector $\eta(\kappa)$ and yield the following proposition. 
	
	\begin{proposition}\label{prop:multi}
		Let $\k\in \R^{6n}_{>0}$, consider  $\eta(\k)\in \R^{3n}_{>0}$ using \eqref{eq:assembly}, and the polynomial $p_{\eta(\k)}(x_1,x_2,x_3)$ in \eqref{eq:p}. 
It holds:
		\begin{enumerate}[(i)]
			\item If $p_{\eta(\k)}(x_1,x_2,x_3)>0$ for all $(x_1,x_2,x_3)\in \R^3_{>0}$, then $\k$ precludes multistationarity.
			\item If $p_{\eta(\k)}(x_1^*,x_2^*,x_3^*)<0$ for some $(x_1^*,x_2^*,x_3^*)\in \R^3_{>0}$, then $\k$ enables multistationarity. In this case the stoichiometric compatibility class containing $\varphi_\k\big(x_1^*,\frac{x_1^*}{x_2^*},x_3^*\big)$ contains at least two positive steady states. 
		\end{enumerate}
		
	\end{proposition}

	In Appendix~\ref{sec:appendix} we explicitly compute $p_{\eta}$ as the determinant of the matrix $J$ in \eqref{matrix:Jac2}. We now explore the structure of the polynomial $p_\eta(x_1,x_2,x_3)$ using techniques from real algebraic geometry, in particular, via studying its Newton polytope. The polynomial is quadratic in $x_3$ and can be written as:
\begin{equation}\label{eq:p_decomp}
p_\eta(x_1,x_2,x_3) = A_2(x_2) x_3^2+ \big(A_{10}(x_2) + A_{11}(x_2) x_1\big)\, x_3+ \big(A_{00}(x_2) + A_{01}(x_2)x_1\big), 
\end{equation}
	with
		\begin{small}
		\begin{align*}
		A_2 (x_2) &= \left(1+\sum_{i=0}^{n-1} c_i \, x_2^{i+1}  \right) \left(\sum_{\ell=1}^{2n-3} \sum_{i+j=\ell} (i-j)\, a_i \, b_j \, x_2^{\ell +1} \right),  \\  A_{11}(x_2) & = (1+x_2) \left( \sum_{\ell = 1 }^{2n-3} \sum_{i+j=\ell}  (i- j) a_i \,b_j\,   x_2^{\ell}    \right), \nonumber\\
		A_{10}(x_2) &=\sum_{\ell=0}^{n-1} (\ell+1)\, a_\ell \, x_2^\ell  -  \sum_{{\ell=1}}^{n-1} \ell \,b_\ell \,  x_2^{\ell +1} 
		+   \sum_{\ell = 0 }^{2n-2} \sum_{i+j=\ell}   (j+1-i) \,b_i\, c_j \,   x_2^{\ell +2} +   \sum_{\ell = 1 }^{2n-3} \sum_{i+j=\ell}   (i-j) \,a_i\, c_j \,   x_2^{\ell+1}, \nonumber\\
		A_{01}(x_2) &=   \sum_{i=0}^{n-1} (a_i+ b_i )  x_2^{i}, \qquad \text{ and } \qquad A_{00}(x_2) = 1+\sum_{i=0}^{n-1} c_i \, x_2^{i+1}.\nonumber
		\end{align*}
	\end{small}
	
		\begin{lemma}\label{prop:edgespositive}
For the polynomial $p_\eta$ written as in \eqref{eq:p_decomp}, the polynomials $A_{01}(x_2)$ and $A_{00}(x_2)$  have only positive coefficients, while $A_2(x_2), A_{10}(x_2),$ and $A_{11}(x_2)$ have both positive and negative coefficients. 

Moreover,	$A_2(x_2)$ is nonnegative if and only if $A_{11}(x_2)$ is nonnegative.
	\end{lemma}
	\begin{proof}
		The first part follows easily from inspection of the coefficients of the polynomials. 
The second statement is immediate as 
		both $A_2$ and $A_{11}$ are positive multiples of {$\sum_{\ell = 1 }^{2n-3} \sum_{i+j=\ell}  (i- j)\, a_i \, b_j\,   x_2^{\ell}$}.
			\end{proof}
	
	Note that 
		\begin{equation}\label{eq:express}
	\sum_{i+j=\ell} (i-j)\, a_i \, b_j \, x_2^{\ell+1}= \sum\limits_{\substack{i>j\\i+j=\ell}} (i-j)(a_i b_j-a_jb_i)x_2^{\ell+1}. 
	\end{equation}
	Furthermore, recalling the definition of 	$a_i,b_i$, and $c_i$ in \eqref{eq:abc}, and of $T_i$ in \eqref{eq:T}, 	 we obtain that 
	\begin{equation} \label{eq:abab}
	\begin{aligned}
	a_ib_j-a_jb_i &=K_i T_{i-1} L_j T_{j}-K_j T_{j-1} L_i T_{i}	=T_{i-1}T_{j-1}\big(K_iL_j \tfrac{\k_{6j+3}K_j}{\k_{6j+6}L_j}-K_jL_i\tfrac{\k_{6i+3}K_i}{\k_{6i+6}L_i}\big)\\
	&=T_{i-1}T_{j-1}K_iK_j\big(\tfrac{\k_{6j+3}}{\k_{6j+6}}-\tfrac{\k_{6i+3}}{\k_{6i+6}}\big). 
	\end{aligned}
	\end{equation}
Therefore, the signs of  $a_i b_j-a_jb_i$ depend on the minors of the following matrix:
	\begin{equation}\label{kappamatrix} M_{\k}= { \begin{bmatrix}
		\k_3 & \k_9 & \ldots & \k_{6n-3} \\
		\k_6 & \k_{12} & \ldots & \k_{6n}
		\end{bmatrix}.
	}
	\end{equation}

	\medskip
	\paragraph{\bf Newton Polytope of $p_\eta(x_1,x_2,x_3)$.}
	To analyze the possible signs that $p_\eta(x_1,x_2,x_3)$ attains over the positive orthant, we begin by exploring the properties of its Newton polytope. 
	The rest of this section focuses on the shape and structure of $\newton{p_\eta}.$  	
	We start by noting that 
	$p_\eta$ is linear in $x_1$, and hence can be expressed as 
	\[ p_\eta(x_1,x_2,x_3)= P_0 (x_2,x_3) +x_1 P_1(x_2,x_3) \]
	with 
\begin{equation}\label{eq:p0p1}
P_0(x_2,x_3)=A_2(x_2) x_3^2+A_{10}(x_2) x_3 +A_{00}(x_2)\quad \textrm{and}\quad P_1(x_2,x_3)=A_{11}(x_2)x_3+A_{01}(x_2).
\end{equation}
	Therefore, the two hyperplanes $H_0:= \{x\in \R^3 : x_1=0\}$ and $H_1:= \{x\in \R^3 : x_1=1\}$ contain all exponent vectors, and further
		\begin{equation}\label{eq:decomp}
	\newton{p_\eta} \medcap H_0 = \{0\}\times \newton{P_0},\qquad \newton{p_\eta} \medcap H_1 = \{1\}\times \newton{P_1}.
	\end{equation}

	\begin{theorem}\label{thm:NPverts}
Let $n\geq 1$, and $\eta\in \R^{3n}_{>0}$ such that $a_1b_0-a_0b_1 \neq 0$ and $a_{n-1} b_{n-2} - a_{n-2} b_{n-1} \neq 0$. 
Then 		
			the set of vertices of the Newton polytope $\newton{p_\eta}$ of   $p_\eta$ in \eqref{eq:p_decomp} consists of the following $10$ points:
		\[ \Big\{(0,0,0),(0,n,0),(0,0,1),(0,2n,1),(0,2,2),
		(0,3n-2,2),(1,0,0),(1,n-1,0),(1,1,1),(1,2n-2,1)\Big\}.\]%
Furthermore, $\newton{P_0}=\newton{p_\eta}\medcap H_0$ is  a hexagon with set of vertices 
		\[ \big\{(0,0),(n,0),(0,1),(2n,1),(2,2),(3n-2,2)\big\}.\] 
		\end{theorem}
	
	\begin{proof}
The hypothesis implies that the polynomials $A_2, A_{11}, A_{10}, A_{01}, A_{00}$ in $x_2$ have degree $3n-2, 2n-2, 2n,  n-1, n$ respectively, and lowest exponent $1, 1,0,0,0$ respectively. 
	
				Considering that the vertices of
		$\newton{p_\eta}$ are contained in the union of hyperplanes $H_0 \medcup H_1$,   the set of vertices of $\newton{p_\eta}$
		is the union of the set of vertices of $\newton{p_\eta}\medcap H_0$ and  of $\newton{p_\eta}\medcap H_1$. 
		By \eqref{eq:decomp}, it is enough to find the vertices of 
		$\newton{P_0}$ and $\newton{P_1}$, which are  planar polytopes.
			
As $P_0$ is quadratic in $x_3$, see \eqref{eq:p0p1}, the vertices of $\newton{P_0}$ are on the lines $L_i:=\{(x_2,x_3)\in \R^2 : x_3=i\}$, for $i=0,1,2$.
The polytopes $\newton{P_0}\medcap L_i$ for $i=0,1,2$ correspond to the Newton polytopes of the univariate polynomials $A_{00}, A_{10},$ and $A_{2}$ respectively. 
Their vertices are simply the highest and lowest exponents of these polynomials.
		Therefore, the set of vertices of $\newton{P_0}$ is contained in $S_1=\big\{(0,0),(n,0),(0,1),(2n,1),(2,2),(3n-2,2)\big\}.$ To show that this is exactly the set of vertices of $\newton{P_0}$, it is enough to verify that
		$(0,1),(2n,1)$ do not lie in the convex hull of the $4$ remaining points. 
		This is easily checked by considering the relative position of $(0,1)$ with respect to the line joining $(0,0)$ and $(2,2)$, and similarly, that of $(2n,1)$ with respect to the line joining $(n,0)$ and $ (3n-2,2).$ The Newton polytope $\newton{P_0}$ is therefore the hexagon with vertex set $S_1.$

		For $\newton{P_1}$, we note that $P_1$ in  \eqref{eq:p0p1} is linear in $x_3.$ Hence, as above, we consider first the vertices of $\newton{P_1}\medcap \{(x_2,x_3)\in \R^2 : x_3=i\}$ for $i=0,1$, which are $0,n-1$ and $1,2n-2$ respectively. It follows that the set of vertices of $\newton{P_1}$ is  $S_2=\big\{(0,0),(n-1,0),(1,1),(2n-2,1)\big\}.$
		By \eqref{eq:decomp}, we conclude that the set of vertices of $\newton{p_\eta}$ is as given in the statement.
	\end{proof}

\begin{remark}\label{rk:original}
In the original parameter vector $\k\in \R^{6n}_{>0}$, the condition in Theorem~\ref{thm:NPverts} corresponds 
 by \eqref{eq:abab} to $\k_3\k_{12} - \k_6\k_9\neq 0$ and $\k_{6n-9}\k_{6n} - \k_{6n-3} \k_{6n-6}\neq 0$.
\end{remark}

Given a polynomial $P\in \R[x_1,\ldots, x_m]$, and a proper face $F$ of   $\newton{P}$, we let $P_F$ denote the \emph{restriction} of $P$ to $F$, that is, $P_F$ is the sum of the terms $c_\alpha x_1^{\alpha_1} \cdots x_m^{\alpha_m}$ of $P$ satisfying $(\alpha_1,\dots,\alpha_m)\in F$. 
	In what follows we will use the following well-known result that gives a connection between the signs $P$ attains over $\R^m_{>0}$, and the signs $P_F$  attains.
	
	\begin{proposition}\label{prop:facerestriction}
		Let $F$ be a proper face of the Newton polytope $\newton{P}$ of a polynomial $P\in \R[x_1,\ldots, x_m]$ and let $P_F$ be the restriction of $P$ to $F$. Then for any $x\in \R^m_{> 0}$ such that $P_F(x)\neq 0$, 
			there exists $y\in \R^m_{> 0}$ such that
		\[ \sign(P(y))= \sign(P_F(x)).\]
		In particular, for $t>0$ large enough and $\omega\in \R^m$ in the open outer normal cone of $F$,  $y= (x_1 t^{\omega_1}, \dots, x_m t^{\omega_m})$ generates the desired point. 
	\end{proposition}
	
	A proof of this proposition can be found in \cite[{Prop. 2.3}]{FKWY}.  	Recall that the open outer normal cone of $F$ consists of the vectors $\omega$ such that $\omega \cdot x^\top > \omega \cdot y^\top$ for every $x\in F$ and $y\in \newton{P}\setminus F$ where $x^\top,y^\top$ denotes the transpose. 
		 We  can easily deduce that if the coefficient of the monomial supported at one of the vertices is negative, then there exists some $x\in \R^m_{> 0}$ such that the polynomial is negative.

	\begin{theorem}\label{thm:restricthexa}
		Fix $\eta\in \R^{3n}_{>0}$. 
		\begin{enumerate}[(i)]
			\item The polynomial $p_\eta$ attains a negative value over $\R^3_{>0}$ if and only if 
			$P_0$ attains negative values over $\R^2_{>0}$.
			\item Moreover, $p_{\eta}$ is either positive over the positive orthant or $p_{\eta}$ attains  negative values over $\R^3_{>0}$. In other words, $p_{\eta}$ has a zero in $\R_{>0}^3$ if and only if it takes  negative values.
		\end{enumerate}
	\end{theorem}
	\begin{proof}
The two statements in the theorem  follow from these two implications:
\begin{align*}
 p_\eta(x_1,x_2,x_3)\leq 0  \textrm{ for some }(x_1,x_2,x_3)\in \R^3_{>0}  & \quad \Rightarrow \quad  
P_0(x_2,x_3)<0   \textrm{ for some } (x_2,x_3)\in \R^2_{>0}  \\ & \quad \Rightarrow \quad  p_\eta(x_1,x_2,x_3)< 0  \textrm{ for some }(x_1,x_2,x_3)\in \R^3_{>0}.
\end{align*}

To show the second implication, note that  $\newton{p_\eta} \medcap H_0$   and $\newton{p_\eta} \medcap H_1$ are two facets of $\newton{p_\eta}$, and $p_\eta$ restricted to these facets is respectively $P_0$ and $x_1 P_1$.
Therefore, if $P_0(x_2,x_3)<0$ for some $(x_2,x_3)\in \R^2_{>0}$, then Proposition~\ref{prop:facerestriction} readily gives that $p_\eta$ also attains negative values. 

For the first implication,  $p_\eta(x_1,x_2,x_3)\leq 0$ implies that  $P_0(x_2,x_3)$ or $P_1(x_2,x_3)$ are negative. If $P_1(x_2,x_3)<0$, then by Lemma~\ref{prop:edgespositive} necessarily $A_{11}(x_2)<0$ and hence, by the same lemma, we also have $A_2(x_2)<0$.  As $A_2(x_2)$ is the leading coefficient of $P_0$, there exists $x_3^*>0$ such that $P_0(x_2,x_3^*)<0$.
This concludes the proof of the theorem. 
	\end{proof}
	
As the signs $p_\eta$ attain are completely determined by the signs of $P_0$ over $\R^2_{>0}$ by Theorem~\ref{thm:restricthexa}, we rewrite $P_0$ as $P_\eta$, to indicate dependence on the parameter vector $\eta\in \R^{3n}_{>0}$:
	\begin{align} \label{eqn:relevantpol}
P_\eta(x_2,x_3)  &= A_2(x_2)x_3^2+A_{10}(x_2) x_3+A_{00}(x_2).
	\end{align}

	Theorem~\ref{thm:restricthexa} tells us that the silent scenario in Proposition~\ref{prop:multi}, namely when $p_\eta$ is nonnegative over $\R^3_{>0}$, cannot occur. This gives the following characterization of multistationarity of the $n$-phosphorylation cycle. 
	
\begin{corollary}\label{cor:crit}
Let $\k\in \R^{6n}_{>0}$, consider  $\eta(\k)\in \R^{3n}_{>0}$ from \eqref{eq:assembly}, and   $P_{\eta(\k)}(x_2,x_3)$ as in \eqref{eqn:relevantpol}. 
Then $\k$  enables multistationarity if and only if $P_{\eta(\k)}(x_2^*,x_3^*)<0$ for some $(x_2^*,x_3^*)\in \R^2_{>0}$. 
\end{corollary}

	\section{ \bf Sufficient conditions for multistationarity} \label{section:sufficient}

	In this section we  start by providing conditions on the parameters to ensure multistationarity using Corollary~\ref{cor:crit}. Then, we describe subsets of the parameter space that preclude multistationarity. First, we make the observation that if the polynomial $A_2(x_2)$ in \eqref{eq:p_decomp} is negative, then $P_{\eta}$ in   \eqref{eqn:relevantpol} attains negative values. This follows immediately from $P_{\eta}$ being quadratic in $x_3$ with $A_2(x_2)$ as the leading coefficient. 
	Recall that the sign of the coefficients of $A_2(x_2)$  are determined by \eqref{eq:express}. 

	\begin{theorem}\label{thm:suf}
		If $\k\in \R^{6n}_{>0}$ is such that 
		\[ \k_3\k_{12} - \k_6\k_9<0\qquad \textrm{or}\qquad \k_{6n-9}\k_{6n} - \k_{6n-3} \k_{6n-6}<0,\]
		then $P_{\eta(\k)}$ attains negative values, and hence $\k$ enables multistationarity.
	\end{theorem}
	\begin{proof} By  Theorem~\ref{thm:NPverts}, the set of vertices of $\newton{P_{\eta(\k)}}$ is $\big\{(0,0),(n,0),(0,1),(2n,1),(2,2),(3n-2,2)\big\},$ see Remark~\ref{rk:original}. 
		The coefficient of the lowest degree term of  $A_2(x_2)$ (i.e.  $x_2^2x_3^2$) is a positive multiple of $a_1b_0-b_0a_1$ by  \eqref{eq:express}, and of $\k_3\k_{12} - \k_6\k_9$ by  \eqref{eq:abab}. Similarly, the coefficient of the highest degree term of  $A_2(x_2)$ (i.e.  
		$x_2^{3n-2} x_3^2$) is a positive multiple of $ \k_{6n-9}\k_{6n} - \k_{6n-3} \k_{6n-6}$. 
If one of these coefficients is negative, then $P_{\eta(\k)}$ attains negative values by Proposition~\ref{prop:facerestriction}, and therefore, by Corollary~\ref{cor:crit}, $\k$ enables multistationarity. 
Note that the other coefficients of $P_{\eta(\k)}$ corresponding to vertices of $\newton{P_{\eta(\k)}}$ are strictly positive. 
	\end{proof}
	
	The expressions in Theorem~\ref{thm:suf} correspond to the first and last maximal minors of the matrix~$M_{\k}$ in~\eqref{kappamatrix}.

	\begin{proposition}
Let $\k\in \R^{6n}_{>0}$ be such that  the matrix $M_{\k}$ in \eqref{kappamatrix} has rank one. 
{Then $\k$ enables multistationarity if and only if  the polynomial $A_{10}(x_2)$ is not nonnegative, that is, attains a negative~value.}
	\end{proposition}
	\begin{proof}
By hypothesis, all maximal minors of $M_{\k}$ vanish, and therefore, $A_2$ is identically zero by \eqref{eq:express} and \eqref{eq:abab}. Therefore, $P_{\eta(\k)}(x_2,x_3)=A_{10}(x_2)x_3+A_{00}(x_2)$.
As $A_{00}$ has only  positive coefficients, $P_{\eta(\k)}(x_2,x_3)$ attains negative values if and only if $A_{10}(x_2)$ does. 
		The statement now follows from Theorem~\ref{thm:restricthexa}.
			\end{proof}

	In the next result we show that network \eqref{eq:networknew} can be  multistationary even if all the minors of the matrix $M_{\k}$ are positive.
	
	\begin{theorem}\label{thm:interiornegative}
		Let $\k\in \R^{6n}_{>0}$ be such that 
		all the minors of $M_{\k}$ are positive. 
		Then, by choosing $K_0=\tfrac{\k_1}{\k_2+\k_3}$ small enough by varying $\k_1$ and/or $\k_2$  and leaving the rest of the entries of $\kappa$ fixed, we obtain a parameter vector  that enables multistationarity.
	\end{theorem}
	
	\begin{proof}
Fix the value of the assembled parameters $K_1,\dots,K_{n-1},L_0,\dots,L_{n-1}$ corresponding to $\k$, and construct
$T_0(K_0),\dots,T_{n-1}(K_0)$  as in \eqref{eq:T} but with $K_0$ treated as a parameter. 
For the one parameter family of vectors
$\xi(K_0) = (T_0(K_0),\dots,T_{n-1}(K_0),K_0,\dots,K_{n-1}, L_0,\dots,L_{n-1})\in \R^{3n}_{>0}$, 
the polynomial $P_{\xi(K_0)}(x_2,x_3)$ becomes a polynomial in $K_0,x_2,$ and $x_3$, which we denote as 
	$Q(K_0,x_2,x_3)$. 
	We prove the theorem by studying the Newton polytope  of $Q(K_0,x_2,x_3)$. 
	
		Note that $K_0$ is a factor with exponent $1$ in each of $a_i,b_i,$ and $c_i$. Consider the exponents of $K_0$ and $x_2$ in  the polynomials $A_{00},A_{10},$ and $A_2$ in \eqref{eq:p_decomp}. Close inspection of each of these show that   $\newton{Q}$ is  contained in the convex hull of the following set of points: 
		\begin{multline*}
		{\Omega:=}
		\{(0,0,0),(1,1,0),\ldots,(1,n,0),(1,0,1),\ldots,(1,n,1), (2,2,1),\ldots,(2,2n,1),\\   (2,2,2),\ldots,(2,2n-2,2),(3,3,2),\ldots,(3,3n-2,2)\}.\end{multline*}
Note that some points could appear with zero coefficient for certain  parameter values. 
		
Consider the point $(1,n,1)\in \Omega$. The term of  $Q$ with monomial $K_0 x_2^n x_3$ is $- (n-1) \,b_{n-1}\,  x_2^{n} x_3$. Hence, the coefficient is negative and    $(1,n,1)\in \newton{Q}$. 
With $\omega =(-2(n+1), 2 ,4) \in \R^3$ and $\gamma=-1$, it holds 
\[ \omega \cdot (1,n,1)^\top + \gamma > 0, \qquad \omega \cdot q^\top + \gamma < 0 \textrm{ for all }q\in \Omega \setminus \{(1,n,1)\}. \]%
Therefore, the hyperplane $\omega\cdot x^\top +\gamma =0$
separates $(1,n,1)$ and $\conv(\Omega \setminus \{(1,n,1)\})$. 
It follows that $(1,n,1)\notin \conv(\Omega \setminus \{(1,n,1)\})$, showing that $(1,n,1)$  is a vertex of $\newton{Q}$.

By Proposition~\ref{prop:facerestriction},  $Q$ attains negative values after appropriately choosing $K_0,x_2,x_3>0$. As $\omega$ belongs to the outer normal cone of $(1,n,1)$ by construction, the point $( t^{-2(n+1)} , t^2, t^4)$ makes $Q$ negative for $t$ large enough. This amounts to letting $K_0$ be small enough, and as   $\k_3$ was fixed in the definition of $T_i(K_0)$, this must be achieved by varying $\k_1$ and/or $\k_2$. 
This concludes the proof.
\end{proof}

\begin{remark}
Theorem~\ref{thm:interiornegative} implies that the conditions in Theorem~\ref{thm:suf} are not necessary for multistationarity. It gives, additionally, a way to construct parameters that enable multistationarity, while satisfying the condition that the negative coefficients of $\newton{P_{\eta(\k)}}$ 
correspond to points in the  interior of $\newton{P_{\eta(\k)}}$.

Although we have shown the existence of values of $K_0$ for multistationarity when the rest of the parameters are fixed, 
the theorem remains true when all parameters but $L_{n-1}$ are fixed. This is shown by considering 
$P_{\eta(\k)}$ as a polynomial in the variables $L_{n-1},x_2$ and $x_3$. Then, the point $(0,n,1)$ is a vertex of the corresponding polytope and the proof follows  similar arguments as the proof of Theorem~\ref{thm:interiornegative}.
\end{remark}

	\section{\bf Sufficient Conditions for Monostationarity}\label{section:circuitcovers}
	
In this section we present a systematic way to construct a sufficient condition for monostationarity for all $n$. To this end, we describe a region $R$ in the parameter space such that the polynomial $P_{\eta(\k)}$ in \eqref{eqn:relevantpol} is positive for any $\k \in \R^{6n}_{>0}$, {and contains parameter values for which $P_{\eta(\k)}$ has negative coefficients}. The approach  is a generalization of the methods that are used in \cite[{\S 2.2}]{FKWY} for $n=2$, and in \cite{Yuruk:Thesis} for $n=3$.

For the particular method we use in this section, it is crucial that the coefficients of the terms that correspond to the vertices of $\newton{P_{\eta(\k)}}$ have positive sign. If this is not the case, then  Theorem \ref{thm:suf} implies that multistationarity is enabled. Therefore, throughout this section we only consider parameters $\k \in \R_{>0}^{6n}$ where
	\begin{equation}
	\label{eq:FirstLastMinorsPositive}
	\k_3\k_{12} - \k_6\k_9 > 0 \qquad \textrm{and}\qquad \k_{6n-9}\k_{6n} - \k_{6n-3} \k_{6n-6} > 0 .
	\end{equation}	
For the cases $n=2$ and $n=3$, the conditions   in \eqref{eq:FirstLastMinorsPositive} are sufficient for all minors of the matrix $M_{\k}$ in \eqref{kappamatrix}  to be positive. 
	
\begin{remark}
For $n=2$, assuming that the minors of the matrix $M_{\k}$ are all positive, then $P_{\eta}$ has only one term  that can have  negative coefficient. The exponent that corresponds to this particular term is contained in the interior of $\newton{P_{\eta}}$. Therefore, as discussed in Remark~\ref{remark:SAGE}, $P_{\eta(\k)}$ is nonnegative if and only if it is a SONC polynomial. 
\end{remark}

In order to find the region $R$, we will restrict to a scenario where all coefficients of  $P_{\eta(\k)}$ corresponding to exponents at the boundary of $N(p_\eta)$ are nonnegative, and hence, only points in the interior of $\newton{P_{\eta(\k)}}$ can have negative coefficient. In particular, this holds when all minors of the matrix $M_{\k}$ are nonnegative, as then all coefficients of $A_2(x_2)$ are nonnegative by  \eqref{eq:abab} and only $A_{10}(x_2)$ contributes to negative terms to $P_{\eta(\k)}$. 
For $n>3$, the conditions in \eqref{eq:FirstLastMinorsPositive} are not enough to guarantee that all minors of $M_{\k}$ are positive. In Proposition~\ref{prop:AllMinorsPos}, we give a sufficient condition on $a_i,b_i,c_i$, {as defined in Equation~\eqref{eq:abc}}, that ensures that all minors of $M_{\k}$ are positive. 
\color{black}

\begin{proposition}
\label{prop:AllMinorsPos}
Assume that for a given $\k\in \R^{6n}_{>0}$, it holds that $a_{i+1}b_{i}-a_{i}b_{i+1}> 0$ (resp. $\geq 0$) for all $i \in  \{0, \dots, n-1\}$. Then, all $2$ by $2$ minors of $M_{\k}$ are positive (resp. nonnegative).
\end{proposition}
\begin{proof}
We need to show that  $\k_{6i+3}\k_{6j+6} - \k_{6i+6}\k_{6j+3}> 0$, equivalently  $\tfrac{\k_{6i+3}}{\k_{6i+6}} - \tfrac{\k_{6j+3}}{\k_{6j+6}} > 0$ for all $j>i$. Note that
	\[ \frac{\k_{6i+3}}{\k_{6i+6}} - \frac{\k_{6j+3}}{\k_{6j+6}}  = \sum_{k=1}^{j-i}  \frac{\k_{6(i+k-1)+3}}{\k_{6(i+k-1)+6}} - \frac{\k_{6(i+k)+3}}{\k_{6(i+k)+6}}.   \] 
By hypothesis and \eqref{eq:abab}, all terms in the sum are positive and hence, so is their sum. The argument for nonnegativity is analogous. This concludes the proof. 
\end{proof}
	
In the rest of this section, we only consider parameter vectors $\k\in \R^{6n}_{>0}$  for which all $2$ by $2$ minors of $M_{\k}$ are nonnegative and \eqref{eq:FirstLastMinorsPositive} holds. These yield to parameter vectors $\eta\in \R^{3n}_{>0}$ such that the set of exponents of the monomials of $P_{\eta}$ that may have negative coefficient is 
\begin{align*}
		A_{-}:=\left\{ (2,1),(3,1),\dots, (2n-2,1) \right\}, 
\end{align*}
and we define 
\begin{equation*}\label{eq:iota}
	\iota_i := (i+1,1), \qquad \text{for}\quad  i\in \{1,\ldots,2n-3\}.
\end{equation*}
Furthermore, we know from Theorem~\ref{thm:NPverts} that $\newton{P_{\eta}}$ is a hexagon for $n\geq 2$, and let us label its vertices as follows:
\begin{equation*}
		\label{Equation:DefinitionAlphas}
		\alpha_1 := (0,0), \quad \alpha_2 := (2,2), \quad \alpha_3 := (2n,1) , \quad \alpha_4 := (0,1), \quad \alpha_5 := (n,0), \quad \alpha_6 := (3n-2,2).
		\end{equation*}
{Under the current assumptions, the coefficients corresponding to these vertices are nonnegative. 
The triangles
\begin{equation}\label{eq:delta}
\Delta_1:= \conv(\{\alpha_1,\alpha_2,\alpha_3\})\qquad \textrm{and}\qquad \Delta_2:= \conv(\{\alpha_4,\alpha_5,\alpha_6\}),
\end{equation}  
contain $\iota_i$  in the relative interior for all $i\in \{1,\ldots,2n-3\}$. 
Hence, our goal now is to construct  two circuit polynomials with supports $\Delta_1\cup \{\iota_i\} $, $\Delta_2\cup \{\iota_i\} $ for each exponent $\iota_i\in A_{-}$, such that we can apply Theorem~\ref{Theorem:CircuitPolynomialNonnegativity} repeatedly to ensure nonnegativity of $P_{\eta}$. }

{Consider $\iota_i \in A_-$, $x=(x_2,x_3)$, and the two polynomials}: 
\begin{align}
\label{eqn:circuitpolys}
	C_{1,i}(x) = \sum_{j=1}^3 \tfrac{\coef{P_{\eta}}{\alpha_j}}{2n-3} x^{\alpha_j} + \delta_{1,i} x^{\iota_i} \quad \text{ and }\quad
	C_{2,i}(x) = \sum_{j=4}^6 \tfrac{\coef{P_{\eta}}{\alpha_j}}{2n-3} x^{\alpha_j} + \delta_{2,i} x^{\iota_i},
\end{align}
where $\delta_{1,i}, \delta_{2,i} \in \R$ are chosen  such that $\delta_{1,i} + \delta_{2,i} =  \coef{P_{\eta}}{\iota_i}$. 
	By construction, $\newton{C_{1,i}}$ and $\newton{C_{2,i}}$ are respectively $\Delta_1$, $\Delta_2$ for all $i\in \{1,\ldots,2n-3\}$
	and hence $C_{1,i}$ and $C_{2,i}$ are  circuit polynomials.

We denote the circuit numbers associated with $C_{1,i}$ and $C_{2,i}$ with $\Theta_{1,i}$ and $\Theta_{2,i}$, respectively. Then, for $i\in \{1,\ldots,2n-3\}$,	it holds
	{\fontsize{10}{10}\selectfont
		\begingroup
		\allowdisplaybreaks
		\begin{align}
		\label{eq:CN1}
		\Theta_{1,i}  &= \left( \frac{\coef{C_{1,i}}{\alpha_1}}{\frac{2n-1-i}{4n-2}} \right)^{\frac{2n-1-i}{4n-2}} 
		\left(\frac{\coef{C_{1,i}}{\alpha_2}}{\frac{2n-1-i}{4n-2}} \right)^{\frac{2n-1-i}{4n-2}} 
		\left(\frac{\coef{C_{1,i}}{\alpha_3}}{\frac{2i}{4n-2}} \right)^{\frac{2i}{4n-2}} \\
		&= \left(\frac{\coef{P_\eta}{\alpha_1}}{(2n-1-i)} \right)^{\frac{2n-1-i}{4n-2}} 
		\left(\frac{\coef{P_\eta}{\alpha_2}}{(2n-1-i)} \right)^{\frac{2n-1-i}{4n-2}} 
		\left(\frac{\coef{P_\eta}{\alpha_3}}{2i} \right)^{\frac{2i}{4n-2}} \left( \frac{4n-2}{2n-3} \right), \nonumber
		\end{align} \endgroup}%
and similarly,
		{\fontsize{10}{10}\selectfont
		\begin{align}
		\label{eq:CN2}
		\Theta_{2,i} & = \left(\frac{\coef{P_\eta}{\alpha_4}}{(4n-2i-4)} \right)^{\frac{4n-2i-4}{4n-2}} 
		\left(\frac{\coef{P_\eta}{\alpha_5}}{(i+1)} \right)^{\frac{i+1}{4n-2}} 
		\left(\frac{\coef{P_\eta}{\alpha_6}}{i+1} \right)^{\frac{i+1}{4n-2}} \left( \frac{4n-2}{2n-3} \right)	. 	
		\end{align}
	}%

For each $i\in \{1,\ldots,2n-3\}$ we define the region $R_i\subseteq \R_{> 0}^{6n}$ as follows:
	 {	\begin{align*}
	 	R_i :=\{\k \in \R_{> 0}^{6n} & :   \coef{P_{\eta(\k)}}{\iota_i} \geq - \Theta_{1,i} - \Theta_{2,i}, \  \\
		 & \quad 	\k_3\k_{12} - \k_6\k_9 > 0,   \ \k_{6n-9}\k_{6n} - \k_{6n-3} \k_{6n-6} > 0, \text{  and}\\ &  \quad
	 	\k_{6\ell+3}\k_{6j+6} - \k_{6\ell+6}\k_{6j+3}\geq  0  \text{ for  all } \ell<j,\ \ell={1},\ldots, n-3\},
	 	\end{align*}}%
where $ \Theta_{1,i}, \Theta_{2,i}$ are defined using $\eta(\k)$. {Note that if $\k \in R_i$, all minors of $M_{\k}$ are nonnegative, and furthermore, the inequalities in \eqref{eq:FirstLastMinorsPositive} hold.} We denote the intersection of these $2n-3$ regions by
		\begin{align}
			\label{eqn:intersectedregions}
			R:=\bigcap\limits_{i=1}^{2n-3} R_i.
		\end{align}
	
\begin{theorem}
\label{Theorem:CircuitMono}\label{Proposition:CircuitNonEmptyIntersection}
With $R\subseteq \R^{6n}_{>0}$ as in \eqref{eqn:intersectedregions}, it holds:
\begin{itemize}
	\item[(i)] There exists $\k\in R$ such that $P_{\eta(\k)}$ has some negative coefficient. In particular, $R\neq \emptyset$.
	\item[(ii)] $P_{\eta(\k)}(x_2,x_3)>0$  for any $\k \in R$ and $(x_2,x_3)\in \R^2_{>0}$, and hence $\k$ precludes multistationarity.
\end{itemize}
\end{theorem}
\begin{proof}
We start by proving (ii). Let $\k \in R$, and $\eta:=\eta(\k)$. Let $\support{P_{\eta}}$ be the set of   exponents of $P_{\eta}$, and define $A_+=\support{P_{\eta}} \setminus \left(A_- \cup \Delta_1\cup \Delta_2 \right)$ (with $\Delta_1,\Delta_2$ as in \eqref{eq:delta}. We rewrite $P_{\eta}$ using the circuit polynomials defined in \eqref{eqn:circuitpolys} as follows:
\begin{align}
		\label{Equation:CircuitDecomp}
		P_{\eta}(x) = \sum_{j=1}^{2} \sum_{i=1}^{2n-3} C_{j,i}(x) + \sum_{\alpha \in A_+}\coef{P_{\eta}}{\alpha}x^{\alpha}.
\end{align}
For any $\alpha \in A_+$, it holds that $\coef{P_\eta}{\alpha} \geq 0$.	Recall that the coefficient of $x^{\iota_i}$ in $C_{j,i}$ equals $\delta_{j,i}$, and these could be chosen arbitrarily as long as $\delta_{1,i} + \delta_{2,i} = \coef{P_\eta}{\iota_i}$. This equality holds if we pick
\[ \delta_{j,i}=-\Theta_{j,i}+\frac{\coef{P_\eta}{\iota_i}+ \Theta_{1,i} + \Theta_{2,i}}{2}, \qquad j=1,2.\]
Furthermore, as in $R$ it holds $\coef{P_\eta}{\iota_i} \geq - \Theta_{1,i} - \Theta_{2,i}$, we find that $\delta_{j,i} \geq -\Theta_{j,i}$ for $j=1,2$ and hence $C_{1,i}$ and $C_{2,i}$ are nonnegative circuit polynomials by Theorem \ref{Theorem:CircuitPolynomialNonnegativity}. By the decomposition \eqref{Equation:CircuitDecomp}, $P_\eta$ is nonnegative as well, and combined with Corollary~\ref{cor:crit}, we have shown (ii). 

\smallskip
To prove (i)  we construct explicitly a point satisfying the conditions. We consider the one parameter family of reaction rate constants $\kappa(k)\in \R^{6n}_{>0}$ with $k>0$ defined as	
\begin{align}
		\label{eqn:oneparamfamily}
		\k_j = \begin{cases} 
		{n-i} & \text{if }  j = 6i+3 \text{ and } i<n, \\
		2 & \text{if }  j = 6i+4 \text{ and } i<n, \\
		k & \text{if } j = 6n-5 = 6(n-1) + 1, \\
		1 & \text{otherwise},
		\end{cases}
\end{align}
and let $\eta(k):=\eta(\k(k))\in \R^{3n}_{>0}$. We show that there exists an interval $I \subset \R_{>0}$ such that for any $k \in I$, $\k(k) \in R$ and $P_{\eta(k)}$ contains a negative term. For any such $\kappa(k)$, we have $L_i = 1$ for any $i \in \left\{0,\dots,n-1\right\}$, and hence, $b_i = c_i =T_i$ for all $i \in \left\{0,\dots,n-1\right\}$. Similarly, we have
\begin{align}\label{eq:T_under_assump}
K_i=\begin{cases} 
\frac{1}{n-i+1} & \text{if } i \in \left\{0, \dots, n-2 \right\}, \\
\frac{k}{2} & \text{if } i = n-1,
\end{cases}  \quad  \text{and} \quad
T_i = \begin{cases} 
\frac{n-i}{n+1} & \text{if } i \in \left\{-1,0, \dots, n-2 \right\}, \\
\frac{k}{n+1} & \text{if } i = n-1. \\
\end{cases}
\end{align}

We start by noting that for each  $\kappa(k)$ and for any $i \in \left\{0, \dots, n-1 \right\}$,	 it holds		
\[\left(\frac{\k_{6{i}+3}}{\k_{6{i}+6}} - \frac{\k_{6(i+1)+3}}{\k_{6(i+1)+6}} \right) = (n-i)-(n-i-1) = 1 > 0.\] 
Hence, all minors of $M_{\k(k)}$ are positive by Proposition~\ref{prop:AllMinorsPos}, and so are the coefficients of $P_{\eta(k)}$ corresponding to the exponents $\alpha_1,\dots,\alpha_6$. It follows from \eqref{eq:CN1} and \eqref{eq:CN2} that  $\Theta_{1,i} (k)+ \Theta_{2,i}(k)  >0$ for all $i\in \{1,\ldots,2n-3\}$, where $\Theta_{1,i}(k)$, $\Theta_{2,i}(k)$ are the circuit numbers   obtained for $\kappa(k)$.

Let $\mathcal{C}_{\ell}(k) := \coef{P_{\eta(k)}}{\iota_{\ell-1}}$, which is the coefficient of $x_2^{\ell}$ in the polynomial $A_{10}(x_2)$. We show in Lemma~\ref{Lemma:AppendixCoefs} in the Appendix, that for any $k>0$, $\mathcal{C}_\ell(k)>0$ when $\ell\in \left\{2,\dots,2n-2\right\}$ and $\ell\neq n$. Thus, $\k(k)\in R_\ell$ for any $\ell\neq n$.

We consider the coefficient of $x_2^n$ in $A_{10}(x_2)$,  see \eqref{eq:p_decomp}. As for $\k(k)$ we have $b_i c_j=T_i T_j$ and $\sum_{ i+j=n-2 }  (j+1-i) T_i\, T_j  = \sum_{ i+j=n-2 }  T_i\, T_j$, we have
\[ \mathcal{C}_{n}(k) = -  (n-1) \,b_{n-1} 
		+    \sum_{ i+j=n-2 }  T_i\, T_j  +   \sum_{\substack{i+j=n-1 \\ i > j}}   (i-j) \left(\,a_i\, c_j 	-\,a_j\, c_i  \right). \]
Using the identities \eqref{eq:abc}, \eqref{eq:T_under_assump}, and \eqref{eq:ac_differences} we find
$a_{n-1}c_0-a_0c_{n-1}= \tfrac{k\, (n-1)}{(n+1)^2}$ and $b_{n-1}=\tfrac{k}{n+1}$. This gives:

$$ \mathcal{C}_{n}(k) = k \left(\frac{n-1}{n+1}\right) \left(\frac{n-1}{n+1} - 1  \right)  +  \sum_{ i+j=n-2 }  T_i\, T_j \,  +  \sum_{\substack{i+j=n-1 \\ i > j,  i \neq n-1}}   (i-j) \left(\,a_i\, c_j -\,a_j\, c_i  \right).$$
The variable $k$ does not appear in the last two summands of $\mathcal{C}_{n}(k)$ and hence, 
$\mathcal{C}_{n}(k)$ is a linear function of $k$ with a negative leading term and a positive constant term. 
		
Let $k_0>0$ be the zero of $\mathcal{C}_n(k)$, and consider the function $g(k)=\mathcal{C}_n(k)+\Theta_{1,n-1}(k) + \Theta_{2,n-1}(k)$. 
Note that $g(k_0)>0$ as $\mathcal{C}_n(k_0)=0$. As $g$ is continuous in $k$, there exists   $\epsilon_0>0$ such that for all $k\in (k_0,k_0+\epsilon_0)$, $g(k)>0$. Hence $\kappa(k)\in R_{n-1}$. By definition of $R$, it then follows that $\kappa(k)\in R$ and by construction $\mathcal{C}_n(k)<0$ for any such $k$, showing (i).
\end{proof}

In the next example, we give the conditions to be in $R$ explicitly when $n=3$ and $\k(k)$ is the parameter vector constructed in \eqref{eqn:oneparamfamily} in  the proof of Theorem~\ref{Theorem:CircuitMono}(i).

\begin{example}
Let $n=3$, $\k(k)=(1,1,3,2,1,1,1,1,2,2,1,1,k,1,1,2,1,1)$ be as in \eqref{eqn:oneparamfamily}. The Newton polytope $\newton{P_{\eta(k)}}$   is shown in the right panel of Figure~\ref{Figure:Hexa}. The circuit numbers from  \eqref{eq:CN1} and \eqref{eq:CN2} are
\begin{align*}
\Theta_{1,1} & = 2^{-\frac{32}{10}} \ 3^{-1} \ 5 \  k^{\frac{4}{10}}, &  		\Theta_{2,1}  & = 2^{-\frac{28}{10}} \ 3^{-\frac{16}{10}} \ 5  \ k^{\frac{6}{10}},   \\
			\Theta_{1,2} &= 2^{-\frac{26}{10}} \ 3^{-\frac{16}{10}} \ 5 \ k^{\frac{8}{10}}, & 		\Theta_{2,2} & = 2^{-3} \ 3^{-\frac{16}{10}} \ 5  \ k^{\frac{9}{10}},  \\			
			\Theta_{1,3} &= 2^{-\frac{32}{10}} \ 3^{-\frac{16}{10}} \ 5 \ k^{\frac{12}{10}}, & 	\Theta_{2,3} &=2^{-\frac{44}{10}} \ 3^{-1} \ 5  \ k^{\frac{12}{10}}. 
\end{align*}		
The vectors $\k(k)\in \R^{18}_{>0}$ are in the regions $R_1, R_2$ and $R_3$ in the parameter space  if they satisfy the following inequalities respectively
			\begin{align*}
			\frac{6k+1}{8} & \geq -2^{-\frac{32}{10}} \ 3^{-1} \ 5 \  k^{\frac{4}{10}} \ - \ 2^{-\frac{28}{10}} \ 3^{-\frac{16}{10}} \ 5  \ k^{\frac{6}{10}}, \\
			\frac{3-k}{4} & \geq - 2^{-\frac{26}{10}} \ 3^{-\frac{16}{10}} \ 5 \  k^{\frac{8}{10}} \ - \ 2^{-3} \ 3^{-\frac{16}{10}} \ 5  \ k^{\frac{9}{10}}, \\
			\frac{7k+4}{16} & \geq - 2^{-\frac{32}{10}} \ 3^{-\frac{16}{10}} \ 5 \  k^{\frac{12}{10}} \ - \ 2^{-\frac{44}{10}} \ 3^{-1} \ 5  \ k^{\frac{12}{10}},
			\end{align*}
respectively. As shown in the proof of Theorem~\ref{Proposition:CircuitNonEmptyIntersection}, the first and the third inequalities hold for all $k>0$. The second inequality holds if $k<10.03$, and the coefficient corresponding to the term $\iota_{n-1}$ in $P_{\eta(k)}$ is negative if $k>3$. Therefore, for $10.03>k>3$,  $P_{\eta(k)}$ contains a negative term, but $P_{\eta(k)}$  is positive. Hence multistationarity is precluded for $\k(k)$.
\end{example}

	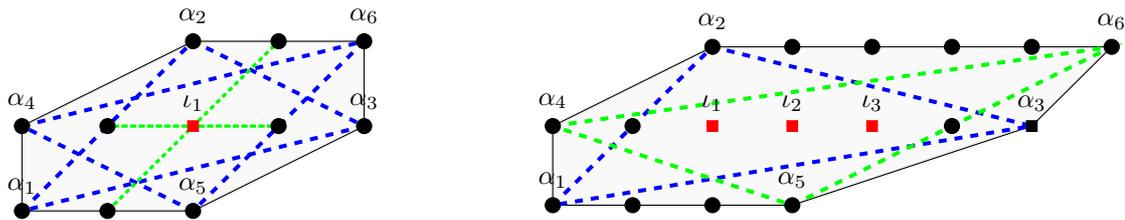
\begin{figure}[t]
			
			\begin{minipage}{.5\textwidth}
				\centering
				\begin{tikzpicture}[scale=0.75]
				\foreach \x in {0,1,2,...,5}
				\foreach \y in {0,1,2} 
				\draw[fill=none,draw=none] (3/2*\x,3/2*\y) circle (1pt) coordinate (m-\x-\y);

				\draw[line width=1.5pt,fill=gray,ultra nearly transparent] (m-4-2) -- (m-2-2) -- (m-0-1) -- (m-0-0) -- (m-2-0) -- (m-4-1) -- cycle;
				\draw  (m-4-2) -- (m-2-2) -- (m-0-1) -- (m-0-0) -- (m-2-0) -- (m-4-1) -- cycle;
				\draw[line width=1.5pt,dashed,blue]  (m-0-0) -- (m-2-2) -- (m-4-1) -- cycle;
				\draw[line width=1.5pt,dashed,blue]  (m-0-1) -- (m-4-2) -- (m-2-0) -- cycle;	
				\draw[line width=1.2pt,dotted,green]  (m-1-0) -- (m-3-2) -- cycle;
				\draw[line width=1.2pt,dotted,green]  (m-1-1) -- (m-3-1) -- cycle;
				
				\node[inner sep=2pt,circle,draw,fill,black,label={$\Vector{\alpha_6}$}] at (m-4-2) {};
				\node[inner sep=2pt,circle,draw,fill,black,label={}] at (m-3-2) {};
				\node[inner sep=2pt,circle,draw,fill,black,label={$\Vector{\alpha_2}$}] at (m-2-2) {};
				\node[inner sep=2pt,circle,draw,fill,black,label={$\Vector{\alpha_4}$}] at (m-0-1) {};
				\node[inner sep=2pt,circle,draw,fill,black,label={}] at (m-1-1) {};
				\node[inner sep=2pt,rectangle,draw,fill,red,label={$\Vector{\iota_1}$}] at (m-2-1) {};	
				\node[inner sep=2pt,circle,draw,fill,black,label={}] at (m-3-1) {};
				\node[inner sep=2pt,circle,draw,fill,black,label={$\Vector{\alpha_3}$}] at (m-4-1) {};

				\node[inner sep=2pt,circle,draw,fill,black,label={$\Vector{\alpha_5}$}] at (m-2-0) {};
				\node[inner sep=2pt,circle,draw,fill,black,label={}] at (m-1-0) {};			
				\node[inner sep=2pt,circle,draw,fill,black,label={$\Vector{\alpha_1}$}] at (m-0-0) {};			
				\end{tikzpicture}
			\end{minipage}%
			\begin{minipage}{.5\textwidth}
				\centering
				\begin{tikzpicture}[scale=0.70]
				\foreach \x in {0,1,2,...,9}
				\foreach \y in {0,1,2} 
				\draw[fill=none,draw=none] (3/2*\x,3/2*\y) circle (1pt) coordinate (m-\x-\y);

				\draw[line width=1.5pt,fill=gray,ultra nearly transparent] (m-7-2) -- (m-2-2) -- (m-0-1) -- (m-0-0) -- (m-3-0) -- (m-6-1) -- cycle;
				\draw  (m-7-2) -- (m-2-2) -- (m-0-1) -- (m-0-0) -- (m-3-0) -- (m-6-1) -- cycle;
				\draw[line width=1.5pt,dashed,blue]  (m-0-0) -- (m-2-2) -- (m-6-1) -- cycle;
				\draw[line width=1.5pt,dashed,green]  (m-0-1) -- (m-7-2) -- (m-3-0) -- cycle;

				\node[inner sep=2pt,circle,draw,fill,black,label={$\Vector{\alpha_6}$}] at (m-7-2) {};
				\node[inner sep=2pt,circle,draw,fill,black,label={}] at (m-6-2) {};
				\node[inner sep=2pt,circle,draw,fill,black,label={}] at (m-5-2) {};
				\node[inner sep=2pt,circle,draw,fill,black,label={}] at (m-4-2) {};
				\node[inner sep=2pt,circle,draw,fill,black,label={}] at (m-3-2) {};
				\node[inner sep=2pt,circle,draw,fill,black,label={$\Vector{\alpha_2}$}] at (m-2-2) {};
				\node[inner sep=2pt,circle,draw,fill,black,label={$\Vector{\alpha_4}$}] at (m-0-1) {};
				\node[inner sep=2pt,circle,draw,fill,black,label={}] at (m-1-1) {};
				\node[inner sep=2pt,rectangle,draw,fill,red,label={$\Vector{\iota_1}$}] at (m-2-1) {};	
				\node[inner sep=2pt,rectangle,draw,fill,red,label={$\Vector{\iota_2}$}] at (m-3-1) {};
				\node[inner sep=2pt,rectangle,draw,fill,red,label={$\Vector{\iota_3}$}] at (m-4-1) {};
				\node[inner sep=2pt,circle,draw,fill,black,label={}] at (m-5-1) {};
				\node[inner sep=2pt,rectangle,draw,fill,black,label={$\Vector{\alpha_3}$}] at (m-6-1) {};
				
				\node[inner sep=2pt,circle,draw,fill,black,label={$\Vector{\alpha_5}$}] at (m-3-0) {};
				\node[inner sep=2pt,circle,draw,fill,black,label={}] at (m-2-0) {};
				\node[inner sep=2pt,circle,draw,fill,black,label={}] at (m-1-0) {};			
				\node[inner sep=2pt,circle,draw,fill,black,label={$\Vector{\alpha_1}$}] at (m-0-0) {};			
				\end{tikzpicture}	
			\end{minipage}%
		
		\caption{{\small The support of $P_\eta$ for $n=2,3$, black circles and red squares correspond to the exponents with positive and negative coefficients, respectively. On the left, blue triangles and green line segments denote the Newton polytopes of the circuit polynomials used in \cite[Theorem 3.5]{FKWY}. On the right, blue and green triangles correspond to the simplices $\Delta_1$ and $\Delta_2$ from \eqref{eq:delta}, respectively.}}
		\label{Figure:Hexa}
	\end{figure}
	
We end this section with some remarks. 
\begin{remark}
\label{Remark:CoverComparison}
If the premise of Proposition \ref{prop:AllMinorsPos} holds, then every term of $P_{\eta(\k)}$ with  exponent in $A_+=\support{P_{\eta}} \setminus \left(A_- \cup \Delta_1\cup \Delta_2 \right)$ has a positive coefficient. This means that we can construct further circuit polynomials using the exponents in $A_+$ as vertices and in $A_-$ as  inner terms. For example, a similar approach was taken for the case $n=2$ in \cite[Theorem 3.5]{FKWY}, where  we considered two additional $1$-dimensional circuits given by $\{(1,0),(2,1),(3,2)\}$ and $\{(1,1),(2,1),(3,1)\}$ (green line segments in the left pane in Figure \ref{Figure:Hexa}). These were used to formulate a sufficient condition for monostationarity stronger than  that of Theorem~\ref{Theorem:CircuitMono}. In Theorem~\ref{Theorem:CircuitMono}, we omit these extra polynomials in  order to keep the notation tidy.

The region $R$ in \eqref{eqn:intersectedregions} is a carefully constructed set, and hence Theorem \ref{Theorem:CircuitMono} might only address a limited part of the monostationarity region. 
The ideas presented in this section can be used when $P_{\eta(\k)}$ has an exponent with a negative coefficient that lies on a positive dimensional face $F$ of its Newton polytope, by applying them to the polynomial  restricted to $F$.
\color{black}
\end{remark}

\begin{remark}
\label{Remark:REPandSAGE}
We note that if the values of the reaction rate constant $\k$ is  known, then we can find a lower bound for $P_{\eta(\k)}(x_2,x_3)$ by using various convex and polynomial optimization methods such as solving a relative entropy program induced by the SONC/SAGE decomposition \cite[Example 3.1]{Chandrasekaran:Shah:RelativeEntropy} \color{black} or using semidefinite programming \cite{Blekherman:Parrilo:Thomas}. However, for the results presented here, the parameters are a priori symbolic. Hence, we instead use the method of choosing circuit covers to obtain inequalities on the parameters, which yields symbolic necessary conditions for multistationarity.
\end{remark}

\section{\bf Connectivity}\label{section:connected}

As motivated in the introduction, even in lack of a full description of the parameter region of multistationarity, it is of interest to understand the topological properties of the region. Here we show that the subset of parameters enabling multistationarity and that of those precluding multistationarity are both path connected in the space of reaction rate constants.  

The strategy we follow, which is applicable to other networks, consists of two parts: first, we identify a subregion of the region of interest that is path connected. For multistationarity, it is the subregion given by Theorem~\ref{thm:suf}, making that theorem critical for this section. For monostationarity, we consider the subregion characterized by  $P_\eta$ having all coefficients positive (non-emptiness of this region is a consequence of our results here). Afterwards, we show that any other point in the region can be joined to a point in the selected subregion via a continuous path. This is done by exploiting again properties of the Newton polytope of $P_\eta$, but now viewing the polynomial as a polynomial in two of the entries of $\eta$.

	
\smallskip	
	Let $X \subseteq \R^{6n}_{>0}$ be the set of parameter vectors $\k$ that enable multistationarity. We prove that $X$ is path connected in two steps: First we show in Lemma \ref{lem:connectedvertex} that a subset $Y$ of $X$ is path connected, and then we construct a path in $X$ from any $\eta$ that enables multistationarity to the subset $Y$  in Theorem \ref{thm:multiconnected}.
	
	\begin{lemma}\label{lem:connectedvertex}
		Let $Y \subseteq \R^{6n}_{>0}$ consist of the parameters $\k\in \R^{6n}_{>0}$  such that $\k_{6n-9}\k_{6n}-\k_{6n-6}\k_{6n-3}<0.$ Then, $Y\subseteq X$ is path connected.
	\end{lemma}
	\begin{proof}
		From Theorem~\ref{thm:suf}, we know that $Y \subseteq X.$ The proof that $Y$ is path connected is a straightforward adaptation of Lemma 5.1 and  Remark 5.3 in \cite{FKWY}. 
	\end{proof}

	\begin{theorem} \label{thm:multiconnected}
	For all $n\geq 2$, the set $X \subseteq \R^{6n}_{>0}$ consisting on the parameters $\kappa$ that enable multistationarity is path connected.
	\end{theorem}
	\begin{proof}
Let $\k\in X$. Then by Corollary~\ref{cor:crit}, there exists $x_2,x_3>0$ such that $P_{\eta(\k)}(x_2,x_3)<0$. 
We proceed similarly to the proof of Theorem~\ref{thm:interiornegative} by treating some parameters as variables. 

Fix the value of the assembled parameters $K_0,\dots,K_{n-2},L_0,\dots,L_{n-1},T_0,\dots,T_{n-2}$ corresponding to $\k$, and construct
$T_{n-1}(\k_{6n-3},K_{n-1})$  as in \eqref{eq:T}   with $\k_{6n-3},K_{n-1}$ treated as parameters.
We fix as well $x_2,x_3>0$, and consider the two parameter family of vectors
\[ \xi(\k_{6n-3},K_{n-1}) = (T_0,\dots,T_{n-2},T_{n-1}(\k_{6n-3},K_{n-1}),K_0,\dots,K_{n-1}, L_0,\dots,L_{n-1})\in \R^{3n}_{>0}.\]
The  polynomial $P_{\xi(\k_{6n-3},K_{n-1})}(x_2,x_3)$ becomes a polynomial in $\k_{6n-3},K_{n-1}$, denoted by $\overline{P}(\k_{6n-3},K_{n-1})$.

Note that the parameters $\k_{6n-3},K_{n-1}$ are only  present  in the terms involving $a_{n-1},b_{n-1}$ and $c_{n-1}$, see  \eqref{eq:T} and \eqref{eq:abc}. 
In particular, $a_{n-1}$ depends only on $K_{n-1}$, while $b_{n-1}$ and $c_{n-1}$ depend on the product $\k_{6n-3}K_{n-1}$.
		 Using \eqref{eq:p_decomp}, we  express $A_2,A_{10}$ and $A_{00}$ in terms of $\k_{6n-3}$ and $K_{n-1}$ in the following way, {where the coefficients $r_{i}^{(j)}\in \R$ appearing in the expressions} are independent of $\k_{6n-3}$ and  $K_{n-1}$:
		\begin{small}
			\begin{align*}
			\overline{A}_2 &= \left(1+\sum_{i=0}^{n-1} c_i \, x_2^{i+1}  \right) \left(\sum_{\ell=1}^{2n-3} \sum_{i+j=\ell} (i-j)\, a_i \, b_j \, x_2^{\ell +1} \right)  \\
			&=(r^{(2)}_0+r^{(2)}_1 \k_{6n-3} K_{n-1})(r^{(2)}_2 +r^{(2)}_3K_{n-1}+r^{(2)}_4\k_{6n-3}K_{n-1}) \qquad \qquad \text{with} \quad r^{(2)}_0,r^{(2)}_1,r^{(2)}_3 >0,\\
			\overline{A}_{10} &=\sum_{\ell=0}^{n-1} (\ell+1)\, a_\ell \, x_2^\ell  -  \sum_{\ell=0}^{n-1} \ell \,b_\ell \,  x_2^{\ell +1} 
			+   \sum_{\ell = 0 }^{2n-2} \sum_{i+j=\ell}   (j+1-i) \,b_i\, c_j \,   x_2^{\ell +2} +   \sum_{\ell = 1 }^{2n-3} \sum_{i+j=\ell}   (i-j) \,a_i\, c_j \,   x_2^{\ell+1}  \\
			&=(r^{(10)}_0+r^{(10)}_1 K_{n-1})-(r^{(10)}_2 +r^{(10)}_3\k_{6n-3}K_{n-1}) + (r^{(10)}_4 +r^{(10)}_5\k_{6n-3}K_{n-1}+r^{(10)}_6\k_{6n-3}^2K_{n-1}^2)\\
			&+(r^{(10)}_7 +r^{(10)}_8K_{n-1}+r^{(10)}_9\k_{6n-3}K_{n-1})\qquad \hspace{4cm} \text{with} \quad r^{(10)}_1,r^{(10)}_8 >0,\\
			\overline{A}_{00} &= 1+\sum_{i=0}^{n-1} c_i \, x_2^{i+1}  =(r^{(2)}_0+r^{(2)}_1 \k_{6n-3} K_{n-1}).
			\end{align*}
		\end{small}%
Collecting all the terms together we write $\overline{P}$ as
\[\overline{P}(\k_{6n-3},K_{n-1})=p_0+p_1 K_{n-1}+p_2 \k_{6n-3} K_{n-1}+p_3\k_{6n-3} K_{n-1}^2+p_4\k_{6n-3}^2 K_{n-1}^2,\]
with $p_1=r^{(2)}_0r^{(2)}_3x_3^2+(r^{(10)}_1+r^{(10)}_8)x_3>0$ and $p_3=r^{(2)}_1r^{(2)}_3x_3^2>0$. The Newton polytope of $\overline{P}$ has vertices $(0,0),(0,1),(1,2),(2,2)$. The  terms of $\overline{P}$ that can be negative are all supported on the edge $F$ joining the vertices $(0,0)$ and $(2,2)$. The outer normal cone of this edge if generated by $v=(v_1,v_2)=(1,-1).$
		
By the choice of $\k$, $x_2,x_3$, we know that $\overline{P}$ attains negative values, namely by choosing $\k_{6n-5},\k_{6n-4},\k_{6n-3}$ to be the entries of   the given $\k$. To avoid confusion, we denote them by $\k_{6n-5}^*,\k_{6n-4}^*,\k_{6n-3}^*$  and the corresponding assembled parameter as $K_{n-1}^*$. 
It follows that $\overline{P}_F(\k_{6n-3}^*,K_{n-1}^*)<0$. 
We evaluate $\overline{P}$ at $q_s=(\k_{6n-3}^* s^{v_1},K_{n-1}^* s^{v_2})=(\k_{6n-3}^* s,K_{n-1}^* s^{-1})$ and obtain
\[ Q(s):= s \overline{P}(q_s)=  p_1 K_{n-1}^*+p_3\k_{6n-3}^* (K_{n-1}^*)^2 + s\,  \overline{P}_F(\k_{6n-3}^*,K_{n-1}^*).   \]%
As for $s=1$, $Q(s)$ is a linear function with negative coefficient and $Q(1)<0$, we have that $Q(s)<0$ for all $s\geq 1$. Hence,
consider the parameter vector $\kappa_s$ obtained from $\kappa$ by replacing $\k_{6n-4}^*,\k_{6n-3}^*$ by $\k_{6n-4}\, s,\k_{6n-3}\, s$  such that the new $K_{n-1}$ is $K_{n-1} s^{-1}$. Then $\k_s$ enables multistationarity for all $s\geq 1$. 
As $s$ increases, so does $\k_{6n-3}s$, and eventually $\k_s$ belongs to $Y$. The parameter vectors $\k_s$ thus give a path connecting the point $\k$ in $X$ with a point in $Y$. This concludes the proof. 
	\end{proof}

	The rest of the section focuses on showing that the region of parameters precluding multistationarity is also path connected. We first show that the region of monostationarity is connected in the $3n$ dimensional parameter space given by $T_i,K_i,L_i$ for $i=0,\ldots, n-1.$ 
Written in terms of $T_i,K_i,L_i$ instead of $a_i,b_i,c_i$,  the terms of the polynomial $P_\eta$ in \eqref{eqn:relevantpol} are	
\begin{small}
		\begingroup
		\allowdisplaybreaks
		\begin{align*}
		A_2 (x_2) &=  \left(1+\sum_{i=0}^{n-1} T_i \, x_2^{i+1}  \right) \left(\sum_{\ell=1}^{2n-3} \sum\limits_{\i+j=\ell} (i-j)\, K_i L_j T_{i-1} \, T_j   x_2^{\ell +1} \right), \\
		A_{10}(x_2) &=
			\sum_{\ell=0}^{n-1} (\ell+1)\, K_\ell T_{\ell-1} \, x_2^\ell  -  \sum_{\ell=0}^{n-1} \ell \,L_\ell T_{\ell} \,  x_2^{\ell +1} 
		+   \sum_{\ell = 0 }^{2n-2} \sum_{i+j=\ell}   (j+1-i) \,L_i T_{i}\, T_j \,   x_2^{\ell +2} \\
		&+   \sum_{\ell = 1 }^{2n-3} \sum\limits_{ \i+j=\ell}  (i-j) K_i T_{i-1}\, T_j    x_2^{\ell+1}, \\
		A_{00}(x_2) &= 1+\sum_{i=0}^{n-1} T_i \, x_2^{i+1}.  \nonumber
		\end{align*}
		\endgroup
	\end{small}
	
	In Theorem \ref{thm:monoconnected} we show that the region $Z\subseteq \R^{3n}_{>0}$ of parameters for which $P_\eta$ is positive  is path connected. We prove this by showing that any point in $Z$ is path connected to a point $\eta \in Z$ such that all the coefficients of $P_{\eta}$ are positive. We then show that any two such points are also path connected in~$Z.$

	\begin{theorem}\label{thm:monoconnected}
		For all $n\geq 1$, the set
				\[ Z:= \{\eta \in \R^{3n}_{>0} : P_\eta(x_2,x_3)>0  \textrm{ for all }(x_2,x_3)\in \R^2_{>0}\} \]
		 is path connected. 	 	\end{theorem}
	\begin{proof}
Let $\eta\in Z$. Hence, for all $x_2,x_3>0$, $P_\eta(x_2,x_3)>0$. 
Fix $(x_2,x_3)\in \R^2_{>0}$ and all parameters in $\eta$ but $K_{n-1},L_0$, and view $P_{\eta}(x_2,x_3)$ as a polynomial in $K_{n-1}$ and $L_0$ with real coefficients. Let $P'$ denote this polynomial. 
It is straightforward to see using \eqref{eq:abab} that $P'$ 	takes the form
		\[ P'(K_{n-1},L_0):=r_3K_{n-1}L_0+r_2K_{n-1}+r_1L_0+r_0, \qquad \textrm{with }r_1,r_2,r_3>0.\]
As in the proof of Theorem~\ref{thm:multiconnected}, we denote by  $K_{n-1}^*,L_0^*$ the corresponding  parameter entries in $\eta$, such that $P'(K^*_{n-1},L^*_0)>0$.

We evaluate $P'$ at $q_s=(K_{n-1}^* s,L_0^* s)$:
\begin{align*}\label{eqn:polQ}
Q(s):= P'(q_s)=  r_3K_{n-1}L_0 s^2+(r_2K_{n-1}+r_1L_0) s+r_0.
\end{align*}
The polynomial $Q$ is quadratic in $s$ with positive leading term and $Q(1)>0$. By Descartes' rule of signs, $Q$ has at most one positive root and hence $Q(s)>0$ for all $s\geq 1$.

As this holds for any fixed $(x_2,x_3)\in \R^2_{>0}$, it follows that 
\[ \eta_s: =(T_0,\ldots,T_{n-1},K_0,\ldots,K_{n-2},K_{n-1}^* s ,L_0^*s ,L_1,\ldots,L_{n-1}) \in Z,\quad \textrm{ for all }s\geq 1.\]

We now show that for $s$ large enough, all coefficients of $P_{\eta_s}$ in $x_2,x_3$ are positive.
The factor of $A_2(x_2)$  that can be negative is $\left(\sum_{\ell=1}^{2n-3} \sum_{i+j=\ell} (i-j)\, K_i L_j T_{i-1} \, T_j \, x_2^{\ell +1} \right).$ For the coefficient of $x_2^k$, seen as a polynomial in all parameters, the coefficients of all the monomials involving $K_{n-1}$ and $L_0$ 
are positive, and there is at least one such monomial for every $2\leq k \leq 2n-2$. Hence, 
by  letting $K_{n-1}$ and $L_0$ be sufficiently large, all coefficients become positive.
		An analogous argument holds for $A_{10}(x_2)$ and therefore we conclude that,  for $s$ large enough, all coefficients of $P_{\eta_s}$ in $x_2,x_3$ are positive.

		\smallskip
All that is left is to show that any two parameter points $\eta_1$, $\eta_2$   such that all coefficients of $P_{\eta_i}$ are positive, are connected by a path in $Z$. 
First note that increasing the coordinates of $K_{n-1}$ and $L_0$ in any of them leads to a parameter point also in $Z$ with all coefficients of $P_{\eta_i}$ positive. Second, consider 
the projection $\pi$ outside the entries $K_{n-1}$ and $L_0$, and form the line segment $ \eta'(s):= s \pi(\eta_1) + (1-s) \pi(\eta_2)$ for $s\in [0,1]$. 
Extend $\eta'(s)$ to a parameter vector $\eta(s,K_{n-1},L_0)$ with $K_{n-1}$ and $L_0$ seen as variables. The coefficients of $P_{\eta(s,K_{n-1},L_0)}$ are polynomials in $K_{n-1}, L_0$ with coefficients in $s$, such that the coefficients of the monomials involving $K_{n-1},L_0$ are positive continuous functions of $s$. As $[0,1]$ is compact, there exist $K_{n-1}^*$ and $L_0^*$ such that for any $K_{n-1}\geq K_{n-1}^*$ and $L_0\geq L_0^*$, all the coefficients of $P_{\eta(s,K_{n-1},L_0)}$ are positive for all $s\in [0,1]$. By picking $K_{n-1}^*$ and $L_0^*$  to be larger than the entries of $\eta_1$ and $\eta_2$, the desired path is the composition of the line segments  in $Z$ 
joining the points   $\eta_1$,  $\eta(0,K_{n-1}^*,L_0^*)$, $\eta(1,K_{n-1}^*,L_0^*)$ and $\eta_2$ in this order. This concludes the proof.
	\end{proof}
	
	In Theorem~\ref{thm:monoconnected} above we have shown that the region of monostationarity is path connected in the parameter space given by the parameters $K_i,L_i,$ and $T_i$ for all $i=0,\ldots,n-1$. Next, we extend this result to the $6n$ dimensional parameter space given by $\k\in \R^{6n}_{>0}$.
	
\begin{corollary}\label{cor:monoconnected}
The set of parameters $\k\in \R^{6n}_{>0}$ precluding multistationarity is path connected for all $n$.
\end{corollary}
\begin{proof}
Let $\k,\k'\in \R^{6n}_{>0}$ preclude multistationarity. We have that $\eta(\k), \eta(\k')\in Z$, with $Z$ as in Theorem~\ref{thm:monoconnected} and hence, there is a continuous path 
\[ \xi(s)=(T_0(s),\ldots,T_{n-1}(s),K_0(s),\ldots,K_{n-1}(s),L_0(s),\ldots,L_{n-1}(s)), \quad s\in [0,1] \] connecting them in $Z$. 
Define the points $\kappa(s) \in \R^{6n}_{>0}$  for $i=0,\dots,n-1$, as
\begin{align*}
 \k_{6i+2}(s) & :=\k_{6i+2}, &  \k_{6i+3}(s) & := \k_{6i+6} \, \frac{T_i(s)}{T_{i-1}(s)} \frac{L_i(s)}{ K_i(s)}, &  \k_{6i+1}(s) & := K_i(s) (\k_{6i+2} + \k_{6i+3}(s)), \\
 \k_{6i+5}(s) & :=\k_{6i+5}, &  \k_{6i+6}(s) & := \k_{6i+6}, &  \k_{6i+4}(s) & := L_i(s) (\k_{6i+5} + \k_{6i+6}).
 \end{align*}
Using \eqref{eq:T}, $\frac{T_i(0)}{T_{i-1}(0)} \frac{L_i(0)}{ K_i(0)}=\frac{\k_{6i+3}}{\k_{6i+6}}$ and hence $\kappa(0)=\kappa$. Furthermore, for all $s\in [0,1]$, $\eta(\kappa(s)) = \xi(s)$. 
Therefore, $\kappa(s)$ precludes multistationarity  for all $s\in [0,1]$ and gives a path 
from $\kappa$ to $\kappa(1)$. All that is left is to connect $\kappa'$ with $\kappa(1)$. 
Note that for $i=0,\dots,n-1$,
\begin{multline*}
(\k_{6i+1}(1), \k_{6i+2}(1), \k_{6i+3}(1), \k_{6i+4}(1), \k_{6i+5}(1), \k_{6i+6}(1))= \\
\left( \tfrac{\k'_{6i+1}}{\k'_{6i+2}+ \k'_{6i+3}} \big(\k_{6i+2} +  \k_{6i+6} \tfrac{\k'_{6i+3}}{\k'_{6i+6}}\big), 
\k_{6i+2},  \k_{6i+6} \tfrac{\k'_{6i+3}}{\k'_{6i+6}},\tfrac{\k'_{6i+4}}{\k'_{6i+5}+ \k'_{6i+6}}, \k_{6i+5}, \k_{6i+6}\right)
\end{multline*}
To this end, define a path $\alpha(s)$ with, for $i=0,\dots,n-1$,
\begin{align*}
 \alpha_{6i+2}(s) & :=\k_{6i+2}  \Big(\tfrac{\k'_{6i+2}}{\k_{6i+2}}\Big)^{s} , &  \alpha_{6i+3}(s) & := \k_{6i+3}'   \Big(\tfrac{\k_{6i+6}}{\k'_{6i+6}}\Big)^{1-s}, & 
  \alpha_{6i+1}(s) &:=  \tfrac{\k'_{6i+1}}{\k'_{6i+2}+ \k'_{6i+3}} \big(\alpha_{6i+2}(s)  + \alpha_{6i+3}(s) \big)
  \\
 \alpha_{6i+5}(s) & :=\k_{6i+5}\Big(\tfrac{\k'_{6i+5}}{\k_{6i+5}}\Big)^{s} , &   \alpha_{6i+6}(s) & := \k'_{6i+6}\Big(\tfrac{\k_{6i+6}}{\k'_{6i+6}}\Big)^{1-s} , &   \alpha_{6i+4}(s) & :=  \tfrac{\k'_{6i+4}}{\k'_{6i+5}+ \k'_{6i+6}} \big(\alpha_{6i+5}(s)  + \alpha_{6i+6}(s) \big).
 \end{align*}
We have $\alpha(0)=\xi(1)$ and $\alpha(1)=\k'$. Furthermore, $\eta(\alpha(s))= \eta(\kappa')$ for all $s\in [0,1]$ and hence, $\alpha(s)$ precludes multistationarity. 
Thus, $\alpha$ gives the desired path and concludes the proof.
	\end{proof}

	\appendix
	\section{} \label{sec:appendix}
	In this appendix we compute the determinant of the matrix $J$ in \eqref{matrix:Jac2}, which we rewrite here for convenience:
	\begin{equation*} \label{matrix:Jac3}  {
		J:= \begin{bmatrix}
		1+  \sum_{i=0}^{n-1} (i+1)\, a_i \, x_2^i x_3& - \sum_{i=0}^{n-1} i\, a_i \, x_2^{i+1}  x_3&\sum_{i=0}^{n-1}  a_i  \, x_1x_2^i \\ 
		\sum_{i=0}^{n-1} (i+1)\, b_i  \,x_2^i x_3 & 1 - \sum_{i=0}^{n-1} i \,b_i\,  x_2^{i+1} x_3&\sum_{i=0}^{n-1}  b_i  \, x_1 x_2^{i} \\ 
		-1+\sum_{i=0}^{n-1} (i+1) \, c_i \,  x_1^{-1} x_2^{i+1}  x_3 & -1 - \sum_{i=0}^{n-1} (i+1) \, c_i \, x_1^{-1} x_2^{i+2}  x_3  & 1+\sum_{i=0}^{n-1} c_i \, x_2^{i+1} 
			\end{bmatrix}.
	}
	\end{equation*}
	
	Note that $x_3$ appears with degree one in each of the entries of columns $1$ and $2$ but the third column does not depend on $x_3$. Therefore, the determinant of $J$
	can be written as a quadratic polynomial in $x_3$:
	\[ P:=\det J= A_2x_3^2+A_1x_3+A_0. \]
	Specifically, after noticing that the first two columns can be written as a sum of two vectors, one independent of $x_3$ and the other a multiple of $x_3$, we write the determinant of $J$ as 
	\[ P = \det J_1+ (\det J_2 + \det J_3)x_3 + (\det J_4) x_3^2,\]%
	where
	\begingroup
	\allowdisplaybreaks
	\begin{align*}
	J_1 & = \begin{bmatrix}
	1 &  0 &\sum_{i=0}^{n-1}  a_i  \, x_1x_2^i \\ 
	0 & 1 &\sum_{i=0}^{n-1}  b_i  \, x_1 x_2^{i} \\ 
	-1 & -1  & 1+\sum_{i=0}^{n-1} c_i \, x_2^{i+1} 
	\end{bmatrix}, \\ 
	J_2 & = \begin{bmatrix}
	1  & - \sum_{i=0}^{n-1} i\, a_i \, x_2^{i+1}  &\sum_{i=0}^{n-1}  a_i  \, x_1x_2^i \\ 
	0 & - \sum_{i=0}^{n-1} i \,b_i\,  x_2^{i+1} &\sum_{i=0}^{n-1}  b_i  \, x_1 x_2^{i} \\ 
	-1&  - \sum_{i=0}^{n-1} (i+1) \, c_i \, x_1^{-1} x_2^{i+2}     & 1+\sum_{i=0}^{n-1} c_i \, x_2^{i+1} 		\end{bmatrix},  \\
	J_3 & = \begin{bmatrix}
	\sum_{i=0}^{n-1} (i+1)\, a_i \, x_2^i & 0 &\sum_{i=0}^{n-1}  a_i  \, x_1x_2^i \\ 
	\sum_{i=0}^{n-1} (i+1)\, b_i  \,x_2^i  & 1&\sum_{i=0}^{n-1}  b_i  \, x_1 x_2^{i} \\ 
	\sum_{i=0}^{n-1} (i+1) \, c_i \,  x_1^{-1} x_2^{i+1}   & -1   & 1+\sum_{i=0}^{n-1} c_i \, x_2^{i+1} 		\end{bmatrix},   \\
	J_4 & = \begin{bmatrix}
	\sum_{i=0}^{n-1} (i+1)\, a_i \, x_2^i & - \sum_{i=0}^{n-1} i\, a_i \, x_2^{i+1}  &\sum_{i=0}^{n-1}  a_i  \, x_1x_2^i \\ 
	\sum_{i=0}^{n-1} (i+1)\, b_i  \,x_2^i & - \sum_{i=0}^{n-1} i \,b_i\,  x_2^{i+1} &\sum_{i=0}^{n-1}  b_i  \, x_1 x_2^{i} \\ 
	\sum_{i=0}^{n-1} (i+1) \, c_i \,  x_1^{-1} x_2^{i+1}    &  -\sum_{i=0}^{n-1} (i+1) \, c_i \, x_1^{-1} x_2^{i+2}   & 1+\sum_{i=0}^{n-1} c_i \, x_2^{i+1} 		\end{bmatrix}. 
	\end{align*}
	\endgroup
	
	We proceed to compute each coefficient $A_0,A_1,A_2$ using the matrices $J_i$, $i=1,2,3,4$. 
	
	\medskip
	\noindent
	\paragraph{\textbf{Computing $A_0$}}
	The determinant of $J_1$ is simply	
	\begin{equation*}\label{eqn:coeff0}
	A_0 = 1+\sum_{i=0}^{n-1} c_i \, x_2^{i+1} +  \sum_{i=0}^{n-1} b_i x_1 x_2^{i}+  \sum_{i=0}^{n-1} a_i x_1 x_2^{i}.
	\end{equation*}

	\noindent
	\paragraph{\textbf{Computing $A_2$}}
	We compute now the determinant of $J_4$.  After adding the first column multiplied by $x_2$ to the second column, $\det J_4$ agrees with the determinant of
	\[ J_4'   = \begin{bmatrix}
	\sum_{i=0}^{n-1} (i+1)\, a_i \, x_2^i &  \sum_{i=0}^{n-1}  a_i \, x_2^{i+1}  &\sum_{i=0}^{n-1}  a_i  \, x_1x_2^i \\ 
	\sum_{i=0}^{n-1} (i+1)\, b_i  \,x_2^i &  \sum_{i=0}^{n-1} \,b_i\,  x_2^{i+1} &\sum_{i=0}^{n-1}  b_i  \, x_1 x_2^{i} \\ 
	\sum_{i=0}^{n-1} (i+1) \, c_i \,  x_1^{-1} x_2^{i+1}    &  0   & 1+\sum_{i=0}^{n-1} c_i \, x_2^{i+1} 
	\end{bmatrix}. \] 
	Subtraction of the second column of $J_4'$ multiplied by $x_1/x_2$ from the third column gives 
	that $\det J_4=\det J_4'=\det J_4''$ with 
	\[ J_4''   = \begin{bmatrix}
	\sum_{i=0}^{n-1} (i+1)\, a_i \, x_2^i &  \sum_{i=0}^{n-1}  a_i \, x_2^{i+1}  & 0 \\ 
	\sum_{i=0}^{n-1} (i+1)\, b_i  \,x_2^i &  \sum_{i=0}^{n-1} \,b_i\,  x_2^{i+1} &0  \\ 
	\sum_{i=0}^{n-1} (i+1) \, c_i \,  x_1^{-1} x_2^{i+1}    &  0   & 1+\sum_{i=0}^{n-1} c_i \, x_2^{i+1} 
	\end{bmatrix}. \] 
	Hence, we compute the determinant of $J_4''$. We start by finding the principal minor obtained by removing the last row and column of $J_4''$:
	\begin{align*}
	M &= \left( \sum_{i=0}^{n-1} (i+1)\, a_i \, x_2^i  \right) \left(\sum_{i=0}^{n-1} \,b_i\,  x_2^{i+1} \right ) - 
	\left(\sum_{i=0}^{n-1}  a_i \, x_2^{i+1} \right) \left( \sum_{i=0}^{n-1} (i+1)\, b_i  \,x_2^i  \right) \\ &=
	\left( \sum_{i=0}^{n-1}\sum_{j=0}^{n-1} (i+1)\, a_i \, b_j x_2^{i+j+1} \right ) - 
	\left(\sum_{i=0}^{n-1}  \sum_{j=0}^{n-1}  a_i  (j+1)\, b_j  \,x_2^{i+j +1} \right) \\
	&=
	\sum_{i=0}^{n-1}\sum_{j=0}^{n-1} (i-j)\, a_i \, b_j x_2^{i+j+1}  = 
	\sum_{\ell=1}^{2n-2} \sum_{i+j=\ell} (i-j)\, a_i \, b_j x_2^{\ell +1}= 
	\sum_{\ell=1}^{2n-3} \sum_{i+j=\ell} (i-j)\, a_i \, b_j x_2^{\ell +1},
	\end{align*}
	where in the last equality we note that the term vanished for $\ell=2n-2$.
	From this, we obtain
	\begin{align*}
	A_2= \det J_4 & = \left(1+\sum_{i=0}^{n-1} c_i \, x_2^{i+1}  \right) \left({\sum_{\ell=1}^{2n-3}} \sum_{i+j=\ell} (i-j)\, a_i \, b_j \, x_2^{\ell +1} \right). 
	\end{align*}

	\noindent
	\paragraph{\textbf{Computing $A_1$}}
	We start by finding $\det J_2$. 
	\[ 	\det J_2  = \det \begin{bmatrix}
	1  & - \sum_{i=0}^{n-1} i\, a_i \, x_2^{i+1}  &\sum_{i=0}^{n-1}  a_i  \, x_1x_2^i \\ 
	0 & - \sum_{i=0}^{n-1} i \,b_i\,  x_2^{i+1} &\sum_{i=0}^{n-1}  b_i  \, x_1 x_2^{i} \\ 
	-1&  - \sum_{i=0}^{n-1} (i+1) \, c_i \, x_1^{-1} x_2^{i+2}     & \sum_{i=0}^{n-1} c_i \, x_2^{i+1} 	   
	\end{bmatrix}  + \det  \begin{bmatrix}
	1  & - \sum_{i=0}^{n-1} i\, a_i \, x_2^{i+1}  & 0 \\ 
	0 & - \sum_{i=0}^{n-1} i \,b_i\,  x_2^{i+1} & 0  \\ 
	-1&  - \sum_{i=0}^{n-1} (i+1) \, c_i \, x_1^{-1} x_2^{i+2}     & 1
	\end{bmatrix}.  \]
	The second determinant is simply
	\[ D_2:= - \sum_{i=0}^{n-1} i \,b_i\,  x_2^{i+1} .\]
	For the first determinant $D_1$, 
	adding the first row to the last gives
	\[ D_1 = \det  \begin{bmatrix}
	1  & - \sum_{i=0}^{n-1} i\, a_i \, x_2^{i+1}  &\sum_{i=0}^{n-1}  a_i  \, x_1x_2^i \\ 
	0 & - \sum_{i=0}^{n-1} i \,b_i\,  x_2^{i+1} &\sum_{i=0}^{n-1}  b_i  \, x_1 x_2^{i} \\ 
	0 &  - \sum_{i=0}^{n-1} ( i\, a_i \, x_2^{i+1} +  (i+1) \, c_i \, x_1^{-1} x_2^{i+2})   & \sum_{i=0}^{n-1}  (a_i  \, x_1x_2^i + c_i \, x_2^{i+1}) 	   
	\end{bmatrix},   \]
	and hence 
	\begingroup
	\allowdisplaybreaks
	\begin{align*}
	D_1 &= - \left( \sum_{i=0}^{n-1} i \,b_i\,  x_2^{i+1}  \right)   \left(  \sum_{i=0}^{n-1}  (a_i  \, x_1x_2^i + c_i \, x_2^{i+1}) 	    \right) 
	+  \left( \sum_{i=0}^{n-1}  b_i  \, x_1 x_2^{i}  \right)  \left( \sum_{i=0}^{n-1} ( i\, a_i \, x_2^{i+1} +   (i+1) \, c_i \, x_1^{-1} x_2^{i+2})  \right)  
	\\ 
	&=   \sum_{i=0}^{n-1}  \sum_{j=0}^{n-1}  (j- i) \,  a_j   \,b_i\,  x_1 x_2^{i+j+1}  
	+   \sum_{i=0}^{n-1} \sum_{j=0}^{n-1}  (j+1-i) \,b_i\, c_j \,   x_2^{i+j+2}    
	\\ 
	&=  \sum_{\ell = 1 }^{2n-3} \sum_{i+j=\ell}  (j- i) \,  a_j  \,b_i\,  x_1 x_2^{\ell+1}  
	+   \sum_{\ell = 0 }^{2n-2} \sum_{i+j=\ell}   (j+1-i) \,b_i\, c_j \,   x_2^{\ell +2}.  
		\end{align*}
	\endgroup
	This gives $\det J_2=D_1+D_2$.
	
	We proceed similarly {with} the computation of $\det J_3$: 
	\begin{align*}
	\det J_3  & =\det  \begin{bmatrix}
	\sum_{i=0}^{n-1} (i+1)\, a_i \, x_2^i & 0 &\sum_{i=0}^{n-1}  a_i  \, x_1x_2^i \\ 
	\sum_{i=0}^{n-1} (i+1)\, b_i  \,x_2^i  & 1&\sum_{i=0}^{n-1}  b_i  \, x_1 x_2^{i} \\ 
	\sum_{i=0}^{n-1} (i+1) \, c_i \,  x_1^{-1} x_2^{i+1} & -1   & \sum_{i=0}^{n-1} c_i \, x_2^{i+1} 		 
	\end{bmatrix}  \\ & + \det  \begin{bmatrix}
	\sum_{i=0}^{n-1} (i+1)\, a_i \, x_2^i & 0 &0 \\ 
	\sum_{i=0}^{n-1} (i+1)\, b_i  \,x_2^i  & 1&0 \\ 
	\sum_{i=0}^{n-1} (i+1) \, c_i \,  x_1^{-1} x_2^{i+1}   & -1   & 1 
	\end{bmatrix}.     
	\end{align*}
	The second determinant $C_2$ is
	\[ C_2 =\sum_{i=0}^{n-1} (i+1)\, a_i \, x_2^i . \]
	For the first determinant $C_1$, we have 
	\[ C_1 = \det  \begin{bmatrix}
	\sum_{i=0}^{n-1} (i+1)\, a_i \, x_2^i & 0 &\sum_{i=0}^{n-1}  a_i  \, x_1x_2^i \\ 
	\sum_{i=0}^{n-1} (i+1)\, b_i  \,x_2^i  & 1&\sum_{i=0}^{n-1}  b_i  \, x_1 x_2^{i} \\ 
	\sum_{i=0}^{n-1} (i+1)\, ( b_i  \,x_2^i + \, c_i \,  x_1^{-1} x_2^{i+1})  & 0   & \sum_{i=0}^{n-1}  (b_i  \, x_1 x_2^{i}  +  c_i \, x_2^{i+1} 		)		\end{bmatrix},  \] 
	and hence
	\begin{align*}
	C_1 &=  \left( \sum_{i=0}^{n-1} (i+1) \,a_i\,  x_2^{i}  \right)   \left( \sum_{i=0}^{n-1}  (b_i  \, x_1 x_2^{i}  +  c_i \, x_2^{i+1} )	 \right) 
	-  \left( \sum_{i=0}^{n-1}  a_i  \, x_1 x_2^{i}  \right)  \left( \sum_{i=0}^{n-1} (i+1)\, ( b_i  \,x_2^i + \, c_i \,  x_1^{-1} x_2^{i+1}) \right)  
	\\ 
	&=   \sum_{i=0}^{n-1}  \sum_{j=0}^{n-1}  (i- j)  \, a_i \,b_j\,  x_1 x_2^{i+j}  
	+   \sum_{i=0}^{n-1} \sum_{j=0}^{n-1}  (i-j) \,a_i\, c_j \,   x_2^{i+j+1}    
	\\ 
	&=  \sum_{\ell = 1 }^{2n-3} \sum_{i+j=\ell}  (i-j) \, a_i  \,b_j\,  x_1 x_2^{\ell}  
	+   \sum_{\ell = 1 }^{2n-3} \sum_{i+j=\ell}   (i-j) \,a_i\, c_j \,   x_2^{\ell+1}.  
	\end{align*}
	
	Putting it all together, we get
	\begingroup
	\allowdisplaybreaks
	\begin{align*}
	A_1 &=C_2+D_2  +  C_1+D_1\\ 
	&= \sum_{\ell=0}^{n-1} (\ell+1)\, a_\ell \, x_2^\ell  - \sum_{\ell=0}^{n-1} \ell \,b_\ell \,  x_2^{\ell +1} + 
	\sum_{\ell = 1 }^{2n-3} \sum_{i+j=\ell}  (j- i)\,  a_j    \,b_i\,  x_1 x_2^{\ell+1}    \\
	&
	+   \sum_{\ell = 0 }^{2n-2} \sum_{i+j=\ell}   (j+1-i) \,b_i\, c_j \,   x_2^{\ell +2}
	+ \sum_{\ell = 1 }^{2n-3} \sum_{i+j=\ell}  (i-j) a_i \,  b_j\,  x_1 x_2^{\ell}  
	+   \sum_{\ell = 1 }^{2n-3} \sum_{i+j=\ell}   (i-j) \,a_i\, c_j \,   x_2^{\ell+1}  
	\\ 
	&= \sum_{\ell=0}^{n-1} (\ell+1)\, a_\ell \, x_2^\ell  -  \sum_{\ell=0}^{n-1} \ell \,b_\ell \,  x_2^{\ell +1} 
	+   \sum_{\ell = 0 }^{2n-2} \sum_{i+j=\ell}   (j+1-i) \,b_i\, c_j \,   x_2^{\ell +2} +   \sum_{\ell = 1 }^{2n-3} \sum_{i+j=\ell}   (i-j) \,a_i\, c_j \,   x_2^{\ell+1}  \\
	&\quad + 
	x_1 (1+x_2) \left( \sum_{\ell = 1 }^{2n-3} \sum_{i+j=\ell}  (j- i)\,   a_j \,b_i\,   x_2^{\ell}    \right).
	\end{align*}
	\endgroup

	\section{}
	\label{appendix:Circuit}
	In this part of the appendix, we provide a supplementary lemma to prove Theorem~\ref{Proposition:CircuitNonEmptyIntersection}.

	\begin{lemma}
		\label{Lemma:AppendixCoefs}
		Let  	$\kappa(k)$ as given in \eqref{eqn:oneparamfamily}. 
If $\ell \in \left\{2,\dots,2n-2\right\}$, $\ell\neq n$, then the coefficient $\mathcal{C}_{\ell}(k)$ of $x_2^{\ell}$ in $A_{10}$ is positive for any $k\in \R_{>0}$.  
	\end{lemma}
	
	\begin{proof}
		First, we consider the case $\ell \in \left\{n+1,\dots,2n-2\right\}$, and note that in this case
		\begin{align*}
		\mathcal{C}_{\ell} (k) &= \sum_{i+j=\ell-2}   (j+1-i) \,b_i\, c_j \ +   \sum_{i+j=\ell-1}   (i-j) \,a_i\, c_j \\
		&=\sum_{ i+j=\ell-2 }   T_i\, T_j \  +  \sum_{\substack{i+j=\ell-1 \\ i > j}}   (i-j) \left(\,a_i\, c_j 	-\,a_j\, c_i  \right),
		\end{align*}
as $b_i = c_i =T_i$ for any $i \in \left\{0,\dots,n-1\right\}$ for $\kappa(k)$, and $\sum_{ i+j=\ell-2 }  (j+1-i) T_i\, T_j  = 
\sum_{ i+j=\ell-2 }  T_i\, T_j$.  
	
		The first summand is positive since each $T_i$ is positive. Likewise, the last summand is also positive, because for any $i,j \in \left\{0,\dots,n-1\right\}$ such that $i>j$,  we find using \eqref{eq:T_under_assump}:
		\begin{align}
		\label{eq:ac_differences}
		\,a_i\, c_j 	-\,a_j\, c_i  =K_iT_{i-1}T_j-K_jT_{j-1}T_i = 		
		\begin{cases} 
		\frac{i-j}{(n+1)^2} & \text{if } i \in \left\{1 \dots, n-2 \right\}, \\
		\frac{k}{(n+1)^2}(n-j-1) & \text{if } i = n-1. \\
		\end{cases}
		\end{align} 
		Therefore $\mathcal{C}_{\ell}(k)>0$ for   $\ell \in \left\{n+1,\dots,2n-2\right\}$. 
		
The coefficient $\mathcal{C}_{\ell}(k)$ for $\ell \in \left\{ 2, \dots,n-1\right\}$ is given by		
\begin{align*}
 \mathcal{C}_{\ell}(k) & = 
		(\ell+1)\, a_\ell  -  (\ell-1) \,b_{\ell-1} 
		+    \sum_{ i+j=\ell-2 } \,T_i\, T_j \  +   \sum_{\substack{i+j=\ell-1 \\ i > j}}   (i-j) \left(\,a_i\, c_j 	-\,a_j\, c_i  \right) \\ 
	&=	\left(  -  (\ell-1) \,b_{\ell-1} 
		+    \sum_{ i+j=\ell-2 } \,T_i\, T_j \right)  + (\ell+1)\, a_\ell +   \sum_{\substack{i+j=\ell-1 \\ i > j}}   (i-j) \left(\,a_i\, c_j 	-\,a_j\, c_i  \right),
\end{align*}
where the second equality is simply a reordering of terms. The second and third terms are positive 
by \eqref{eq:ac_differences}. We will show that the first term is also positive. 

We point out that  $T_{n-1}$ does not appear in any sum. For $\ell=2$, the first term is $-  (\ell-1) \,T_{\ell-1} \ + T_{0}^2=T_0^2>0$ and hence, $\mathcal{C}_2(k)>0.$ For $\ell>2$,  
 \eqref{eq:T_under_assump}  gives $T_iT_j=\frac{(n-i)(n-j)}{(n+1)^2}$, and we obtain 
	\begin{align*}			
	\label{eq:small_degrees2}
	(n+1)^2 \sum_{ i+j=\ell-2 } \,T_i\, T_j & = 
 \sum_{i+j=\ell-2 } (n-i)(n-j) \notag \   =  \sum_{i=0}^{\ell-2} (n-i)(n-\ell+2+i)\\
	&=  \sum_{i=0}^{\ell-2} \Big((n+1)-(i+1)\Big)\Big((n-\ell+1)+(i+1)\Big)  \\ & > (\ell-1)(n-\ell+1)(n+1) = 
	(n+1)^2 (\ell-1) b_{\ell-1}.   
	\end{align*} 
This gives $\mathcal{C}_{\ell}(k)>0$ and concludes the proof. 
	\end{proof}

\medskip
\noindent
\paragraph{\textbf{Acknowledgements. }}
{We thank the anonymous referees for their helpful comments.} 
EF acknowledges funding from the Independent Research Fund of Denmark. TdW is supported by the Deutsche Forschungsgemeinschaft (DFG, German Research Foundation) Emmy Noether Programme, grant WO 2206/1-1. OY is funded by the Deutsche Forschungsgemeinschaft (DFG, German Research Foundation) under Germany's Excellence Strategy – The Berlin Mathematics Research Center MATH+ (EXC-2046/1, project ID 390685689, sub-project AA1-9).

\bibliographystyle{plain}

\end{document}